\documentclass[a4paper]{amsart}

\usepackage{amsmath,amsfonts,amssymb,bm}

\usepackage{cancel}
\usepackage{tikz-cd}
\usepackage{float}
\usepackage[all]{xy}

\usepackage[utf8]{inputenc}
\usepackage[english]{babel}

\usepackage{comment}

\usepackage[numbers,sort&compress]{natbib}
\usepackage{graphicx}
\usepackage{appendix}

\usepackage{hyperref}
\hypersetup{
   colorlinks   =  true,
    linkcolor    = cyan,
    citecolor    = red,
     urlcolor	=magenta,     
}

\usepackage{amsthm}

\newtheorem{definition}{Definition}[section]
\newtheorem{lemma}[definition]{Lemma}
\newtheorem{theorem}[definition]{Theorem}
\newtheorem{proposition}[definition]{Proposition}
\newtheorem{corollary}[definition]{Corollary}
\newtheorem{remark}[definition]{Remark}
\newtheorem{example}[definition]{Example}

\usepackage{cite}

           {\nolinebreak \hfill $\Box$ \end{trivlist}}

\def \R{\mathbb{R}}

\numberwithin{equation}{section}

\title[]{Fouling maps and polynomial first integrals from symmetric tensor fields}
\author{Rafael Azuaje$\;^1$, Juan Carlos  Marrero$\;^2$, Edith Padr\'on$\;^2$}

\date{\today}

\begin{document}

\maketitle

\vspace{-20pt}
\begin{center}

{\small\it $\;^1$ Department of Physics, Faculty of Nuclear Sciences and Physical Engineering, Czech Technical University in Prague. Břehová 7, 115 19 Praha 1, Czech Republic}\\
{\small\it e-mail: azuajraf@cvut.cz}
\\[5pt]

{\small\it $\;^2$ ULL-CSIC Geometr\'{\i}a Diferencial y Mec\'anica Geom\'etrica, Departamento de Matem\'aticas, Estad\'{\i}stica e Investigaci\'on Operativa
and Instituto de Matem\'aticas y Aplicaciones (IMAULL)}\\{\small\it University of La Laguna, Spain}\\
{\small\it e-mail: jcmarrer@ull.edu.es, mepadron@ull.edu.es}
\end{center}

\begin{abstract}
Under the framework of time-independent Hamiltonian mechanics on the cotangent bundles $T^*Q$ of the configuration spaces $Q$ of mechanical systems, we introduce the notion of fouling map as a (not necessarily invertible) generalization of the so-called fouling transformations --canonoid transformations preserving configuration coordinates--. We develop a tensorial method for constructing polynomial fouling maps. We show that each such map induces a $(1,1)$-tensor field invariant under the Hamiltonian flow, whose traces of its powers are polynomial constants of motion. For mechanical Hamiltonian functions --the kinetic energy plus the potential energy on a semi-Riemannian configuration space $(Q,g)$--, we completely characterize polynomial bundle maps arising from symmetric $(k+1,0)$-tensor fields and derive the conditions ensuring their fouling nature. Several explicit examples on the Euclidean plane,  on the 2-sphere  and  on  a Riemannian manifold with a Liouville metric illustrate the method.
\end{abstract}


\tableofcontents

\section{Introduction}
\label{sec:intro}

\subsection{Canonoid transformations and first integrals }
In classical mechanics, the well-known canonical transformations play a fundamental role; indeed, they connect solutions of the Hamilton-Jacobi equation with Hamilton's equations of motion.  From a coordinate (local) approach, they are described as (invertible) coordinate transformations that preserve the canonical Poisson bracket \citep{landau1982mechanics, GPS2002,Calkin96}, while from a geometric (global) approach, they are defined as diffeomorphisms on phase space preserving the underlying geometric structure; in particular, for Hamiltonian systems on symplectic manifolds, they are symplectomorphisms on phase space \citep{Asorey1983generalized,carinena1985canonical,azuaje2025canonical}. The concept of canonoid transformations for Hamiltonian systems was introduced in \citep{SC72} as coordinate transformations that preserve the form of the equations of motion; i.e., a coordinate transformation is canonoid if the equations of motion for the new coordinates are Hamiltonian; of course, every canonical transformation is canonoid; however, the notion of a canonoid transformation is more general than that of a canonical transformation since the underlying geometric structure need not be preserved. 

Canonoid transformations have been studied from a geometric perspective in \citep{CR88,CR89,CFR2013}; they are presented as diffeomorphisms on phase space preserving the Hamiltonian structure. Canonoid transformations are closely related to bi-Hamiltonian structures \citep{MM84,KM96,MMSV2002,magri2004eight,FP2003} --alternative Hamiltonian structures describing the same dynamics-- and, consequently, one can find constants of motion associated with each canonoid transformation \citep{CR88,boleantu2013new,azuaje2023canonical}. Remarkably, \citep{azuaje2023canonical} provides a constructive method to find the associated constants of motion for a given canonoid transformation. Indeed, given a canonoid transformation for a Hamiltonian system, it gives rise to a $(1,1)$-tensor field on phase space that remains invariant under the Hamiltonian dynamical flow, so the traces of its powers are constants of motion. A relevant feature of such constants of motion is that they are polynomials with respect to the momenta, provided that the coordinate expression of the canonoid transformation is polynomial with respect to  the momenta (for details, see \citep{azuaje2023canonical}). Polynomial constants of motion (or integrals of motion) have been the object of study for many decades since they are relevant in the study of integrable and superintegrable systems \citep{kolokol1983geodesic,hietarinta1987direct,tsiganov2005family,daskaloyannis2006unified,marquette2007polynomial,miller2013classical} (see \citep{bolsinov2004integrable} for the integrability of geodesic flows).

Canonoid transformations have been further extended beyond symplectic geometry, in fact, they have also been studied for Hamiltonian systems on Poisson manifolds \citep{RS2015}, contact manifolds \citep{azuaje2024scaling} and more recently on locally conformal symplectic manifolds \citep{azuaje2026canonical}. Despite this extensive theoretical development, only a few explicit examples of canonoid transformations are known. In \citep{CR88}, following the theory of generating functions for canonical transformations, a method is proposed to construct canonoid transformations for a given Hamiltonian system from generating functions; however, as noted there, the resulting equations are significantly more difficult than those for canonical transformations. For this reason, attention was restricted to the so-called fouling transformations, a subclass of canonoid transformations that preserve the position coordinates \citep{landolfi2007certain,tempesta2002quantum}. Even then, the corresponding equations remain challenging to solve, and only two explicit examples are presented in \citep{CR88}. 

\subsection{Killing tensors and first integrals of geodesic flows}
On the other hand, it is well known that there is a one-to-one correspondence between symmetric $(k,0)$-tensor fields on a manifold $Q$ and homogeneous functions of degree $k$ (with respect to the momenta) on $T^*Q$. More precisely, for each symmetric $(k,0)$-tensor $T$ on $Q$, there is a function
$
T^{\mathrm{hom}}:T^*Q\longrightarrow\mathbb{R}
$
satisfying
\[
\Delta(T^{\mathrm{hom}})=k\,T^{\mathrm{hom}},
\]
where $\Delta$ denotes the Liouville vector field on $T^*Q$.

Now, let $Q$ be endowed with a semi-Riemannian metric $g$. Under the above correspondence, symmetric Killing  tensor fields are in one-to-one correspondence with homogeneous functions (with respect to the momenta) that are first integrals of the geodesic flow. Recall that a Killing tensor $T$ is a natural generalization of the notion of a Killing vector field and is a tensor on $Q$ characterized by the condition
\[
\operatorname{Sym}(\nabla T)=0,
\]
where $\nabla$ is the Levi-Civita connection of $g$, and $\operatorname{Sym}(\nabla T)$ denotes the symmetrization of $\nabla T$.

In the case $k=1$, if a Lie group acts on $Q$ by isometries, then Noether's theorem states that the associated cotangent lift admits a momentum map whose components are first integrals of the geodesic flow that are linear with respect to the momenta. Under the above correspondence, these first integrals determine Killing vector fields on $Q$.
\subsection{Our contribution: fouling maps and polynomial first integrals of mechanical Hamiltonian systems}
The framework of this paper is Hamiltonian mechanics on symplectic manifolds. Specifically, this paper considers Hamiltonian systems on the cotangent bundle $T^{*}Q$ of a smooth manifold $Q$; it is well known that the cotangent bundle $T^{*}Q$ of each smooth manifold $Q$ has a canonical symplectic structure, which makes it the natural phase space for time-independent Hamiltonian mechanics. {Given a Hamiltonian function $H\in C^{\infty}(T^{*}Q)$, we introduce the notion of fouling maps for $H$ as bundle maps over the identity that preserve (at least locally) the Hamiltonian nature of the Hamiltonian vector field $X_{H}$}. Such a notion generalizes (from a geometric approach) the concept of fouling transformation, where now the bundle map is not necessarily invertible. From a geometric point of view, this means that we replace symplectic geometry by presymplectic geometry. Despite this, as in the case of fouling transformations, we prove that given a fouling map for a  Hamiltonian function $H$ on $T^{*}Q$, it induces a $(1,1)$-tensor field on $T^{*}Q$ whose traces of its powers are constants of motion of the Hamiltonian dynamics determined by $H$ (see Theorem \ref{th2.5}). 

The main aim of this paper is to provide a procedure that employs tensorial objects on $Q$ to obtain fouling maps in $T^*Q$ for mechanical Hamiltonian functions --Hamiltonians representing the total energy of physical systems \citep{Arnold78,AMRC2008}: kinetic energy plus potential energy-- on semi-Riemannian configuration spaces $(Q,g)$; in fact, in this case, we consider fouling maps defined by symmetric $(k+1,0)$-tensor fields on $Q$, and, more generally, finite sums of such maps. One of the main results of this paper (see Theorem \ref{theoremPsi}) is the complete characterization of fouling maps for a given mechanical Hamiltonian function on $T^{*}Q$ that are finite sums of bundle maps over the identity defined by symmetric $(k+1,0)$-tensor fields. The local coordinate expressions of these last maps are polynomials in the momentum coordinates, so the corresponding associated constants of motion are polynomials in the momentum coordinates. In the particular case when the fouling map is induced by a single symmetric $(k+1,0)$-tensor then the first integrals are homogeneous functions with respect to the momenta. So, if the Hamiltonian function is of kinetic type (that is, the potential energy is zero) then we generate a hierarchy of Killing tensors on the semi-Riemannian manifold. This represents an important difference with respect to the previous methods, which yield a single Killing tensor rather than a hierarchy. To show the effectiveness of the previous results, we apply them to several examples.

\subsection{Organization of the paper} This paper is organized as follows: In Section \ref{secmotivation}, we briefly review the construction of homogeneous first integrals of geodesic flows by means of Killing tensors. In Section \ref{sec2}, we introduce the notion of {\it fouling maps} for a given Hamiltonian system; this notion generalizes fouling transformations from a geometric approach; we show how to obtain constants of motion by considering the traces of the powers of an induced $(1,1)$-tensor field.

The fundamental relation between homogeneous bundle maps over the identity and symmetric mixed tensors on the base manifold is described in Section \ref{sec3}. In this section, we characterize polynomial fouling maps for mechanical Hamiltonian functions, i.e., maps whose local expressions in canonical coordinates are polynomial functions with respect to the momenta; from this kind of maps, we obtain polynomial constants of motion. Several examples of polynomial fouling maps for mechanical Hamiltonian functions are presented in the remaining sections. Specifically, Section \ref{sec4} is restricted to Hamiltonian systems on the Euclidean plane, Section \ref{sec5} considers Hamiltonian systems on the 2-sphere (see \citep{carinena2005central,carinena2008harmonic,diacu2012n} for Hamiltonian systems on curved spaces) and Section \ref{sec6} considers a Riemannian  manifold of dimension $2$ with a Liouville metric.

\section{A motivation: homogeneous first integrals of the geodesic flow}
\label{secmotivation}
Let $(Q,g)$ be a semi-Riemannian manifold, and we denote by $X_K\in {\mathfrak X}(T^*Q)$ the geodesic flow of the metric $g$, that is, the Hamiltonian vector field of the kinetic energy $K:T^*Q\to {\Bbb R}$ given by   
$$K(\alpha_q)=
\frac{1}{2}g^*(\alpha_q,\alpha_q)\;\; \mbox{ with }\alpha_q\in T_q^*Q,$$
$g^*:T^*Q\times_QT^*Q\to {\Bbb R}$ being the dual metric of $g$. We recall that for a Hamiltonian function $H:T^*Q\to \mathbb{R},$ the Hamiltonian vector field  $X_H$ on $T^*Q$ is characterized by 
$$i_{X_H}\omega_Q=dH,$$
where $\omega_Q$ is the canonical symplectic $2$-form on $T^*Q.$

 In the literature, one can find some well-known methods to obtain first integrals for $X_K$ that are linear with respect to  the momenta, for instance, to consider Lie group actions by isometries on the manifold $Q$ or to look for Killing vector fields on $Q.$ 

In the first case, if $\phi:G\times Q\to Q$ is an action by isometries, then $K$ is invariant with respect to the cotangent lift action $T^*\phi:G\times T^*Q\to T^*Q$ of $\phi$. Moreover, from Noether Theorem, for each element $\xi$ of the Lie algebra ${\mathfrak g}$ of $G$, we have a  linear (in the momenta) constant of motion $J_\xi:T^*Q\to {\Bbb R}$ for $X_K$  given by $$ J_\xi(\alpha_q)=\alpha_q(\xi_Q(q)), \mbox{ with }q\in Q\mbox{ and }\alpha_q\in T_q^*Q,$$
where $\xi_Q$ is the infinitesimal generator of $\phi$ for $\xi\in {\mathfrak g}.$

On the other hand, in the second method, if $X$ is a vector field on $Q,$ then the linear (in the momenta) function $$f_X:T^*Q\to {\Bbb R},\;\; f_X(\alpha_q)=\alpha_q(X(q))$$ is a first integral for $X_K$ if and only if 
$${\mathcal L}_Xg=0,$$
i.e., $X$ is Killing. 

But, how can we obtain homogeneous (in the momenta) first integrals of degree greater than one? One of these first integrals  is the Hamiltonian itself, which is  quadratic (i.e., a homogeneous first integral of degree 2).

A natural generalization of the Killing vector field method to higher-order first integrals is to consider Killing tensors. We now briefly review this method: 

Let $T:T^*Q\times_Q\dots \times _QT^*Q\to {\Bbb R}$ be a symmetric $(k,0)$-tensor on $Q.$ Then $T$ induces a homogeneous (in the momenta) function $T^{hom}:T^*Q\to \mathbb{R}$ on $T^*Q$  of degree $k$ given by 
\begin{equation}\label{2.0}
T^{hom}(\alpha_{q})=\frac{1}{k!}T_q(\alpha_{q},\ldots^{(k}\ldots,\alpha_{q}).
\end{equation}
Note that if the local expression of $T$ is 
$$T=T^{i_1\dots i_k}\partial_{q^{i_1}}\otimes \dots \otimes\partial_{q^{i_k}}$$
with $T^{i_1\dots i_k}$ local functions on $Q$, 
then the function $T^{hom}$ is the homogeneous polynomial (in the momenta) 
$$
T^{hom}=\frac{1}{k!}T^{i_1\dots i_k}p_{i_1}\dots p_{i_k}.
$$

Since
$$T^{i_{\sigma(1)}\cdots i_{\sigma(k)}}=T^{i_1\cdots i_k},\mbox{ for all }\sigma \in \mathfrak{S}_k,$$
where $\mathfrak{S}_k$ is the group of permutations of order $k$, one deduces that $T^{hom}$ is $k$-homogeneous, that is, 
$\Delta(T^{hom})=kT^{hom},$ with $\Delta$ the Liouville vector field. We recall that the local expression of $\Delta\in {\mathfrak X}(T^*Q)$ is 
\begin{equation}\label{Liouville}
\Delta=p_i\partial_{p_i}.
\end{equation}

Equivalently, we have that 
$$T^{hom}(\lambda \alpha_q)=\lambda^kT^{hom}(\alpha_q), \mbox{ for }\lambda\in \R \mbox{ and } \alpha_q\in T_q^*Q,$$
and, for every $q\in Q$, the map 
$$T_q^*Q\times\cdots ^{(k}\cdots \times T_q^*Q\to \R$$ defined by 
\begin{equation}\label{m1}
\begin{array}{rcl}
(\alpha_q^1,\dots, \alpha_q^k)&\to& \displaystyle\frac{1}{k!}\left((-1)^{k+1} \displaystyle\sum_{i=1}^k T^{hom}(\alpha_q^i) + (-1)^{k} \sum_{i=1, i<j}^k T^{hom}(\alpha_q^i+ \alpha_q^j)+\dots \right.\\&&\left.\displaystyle-\sum_{i=1}^k T^{hom}(\alpha_q^1+\alpha_q^2+\dots +\alpha_q^{i-1}+\alpha_q^{i+1}+\dots +\alpha_q^k ) + 
 T^{hom}(\alpha_q^1+\dots +\alpha_q^k)\right)
 \end{array}
 \end{equation}
is multilinear. In fact, (\ref{m1}) defines  a symmetric tensor of type $(k,0)$ on $Q$ which is $T.$

Therefore, there is a one-to-one correspondence between homogeneous functions on $T^*Q$ of degree $k$ and symmetric $(k,0)$-tensors on $Q$. 

 The following homogeneity properties of Hamiltonian vector fields for basic functions $f\circ \pi_Q$ and linear functions $Y^\ell:T^*Q\to {\Bbb R}$, with $f\in C^\infty(Q)$ and $Y\in {\mathfrak X}(Q)$,  will be useful in several results of this paper. We recall that the linear function  $Y^\ell$ induced by $Y$ is given by
$$Y^\ell (\beta_q)=\beta_q(Y(q)), \mbox{ for all }\beta_q\in T^*Q_q.$$

Note that $Y^\ell=Y^{hom}.$
\begin{proposition}\label{lemma1}
If $T$ is a symmetric tensor of type $(k,0)$, then 
\begin{enumerate}
\item[(a)] $X_{f\circ \pi_Q}(T^{hom})=(i_{df}T)^{hom},$
\item[(b)] $X_{Y^\ell}(T^{hom})=({\mathcal L}_YT)^{hom},$
 \end{enumerate} 
 for all $f\in C^\infty(Q)$ and $Y\in {\mathfrak X}(Q).$
\end{proposition}
\begin{proof}
We recall that the Hamiltonian vector field $X_{f\circ \pi_Q}$ is just the vertical lift of $df$, i.e.
$$(X_{f\circ \pi_Q})_{\alpha_q}(F)=\frac{d}{dt}_{|t=0}(F(\alpha_q+tdf_{q})),$$ for all $F\in C^\infty(T^*Q)$ (see \citep{LR89}). Then 
$$(i_{df}T)^{hom}(\alpha_q)=\frac{1}{(k-1)!}T_{q}(df_q, \alpha_{q},\ldots,\alpha_{q})$$
and
$$\begin{array}{rcl}
\left(X_{f\circ \pi_Q}(T^{hom})\right)(\alpha_q)&=&\displaystyle\frac{d}{dt}_{|t=0}(T^{hom}(\alpha_q + tdf_q))=\frac{1}{k!}\frac{d}{dt}_{|t=0} (T_q(\alpha_q, \dots \alpha_q) + ktT_q(df_q, \alpha_q,\dots ,\alpha_q))\\[8pt]&=&(i_{df}T)^{hom}_q(\alpha_q).
\end{array}$$
Thus, $(a)$ is proved.

On the other hand, for all $X,Y,Z\in {\mathfrak X}(Q),$ since 
$$(X\otimes Y)^{hom}=\frac{1}{2}X^\ell Y^\ell,$$
then 
\begin{equation}\label{lemm1}
({\mathcal L}_Y(X\otimes Z + Z\otimes X))^{hom}=Z^\ell [Y,X]^\ell + X^\ell[Y,Z]^\ell.
\end{equation}

Moreover, $X_{Y^\ell}(Z^{\ell})=[Y,Z]^\ell.$ Indeed, 
\begin{equation}\label{lemm2}
X_{Y^\ell}(Z^{\ell})=\omega_Q(X_{Z^\ell}, X_{Y^\ell})=\{Z^\ell, Y^\ell\}= [Y,Z]^\ell.
\end{equation}
Therefore, from (\ref{lemm1}) and (\ref{lemm2}), 
$$({\mathcal L}_Y(X\otimes Z + Z\otimes X))^{hom}= X_{Y^\ell}((X\otimes Z + Z\otimes X)^{hom}).$$
This proves (b) for $k=2.$ Note that the case $k=1$ follows from (\ref{lemm2}). In a similar way, one may check (b) for $k\geq 3.$ 
\end{proof}

Now, we consider the geodesic flow $X_K$ of the Riemannian manifold  $(Q,g)$.   The previous result, which relates Killing vector fields to first integrals linear in the momenta, admits a natural generalization: the homogeneous function $T^{{hom}}$ is a first integral of $X_K$ if and only if $T$ is a Killing tensor, i.e., 
$$Sym(\nabla T)=0,$$
where $\nabla$ is the Levi-Civita connection of $g$ and $Sym(\nabla T)$ is the symmetrization of $\nabla T.$ Killing tensors have been extensively studied in relation to constants of motion for Hamiltonian systems \citep{thompson1986killing,bolsinov2003geometrical,benenti2016separability}.

In this paper, we develop a new method for deriving homogeneous, and more generally, polynomial first integrals of a given mechanical Hamiltonian system, not necessarily of kinetic type, through the construction of an operator field induced by a fouling map. 

\section{Bundle maps over the identity on the cotangent bundle and Hamiltonian systems}
\label{sec2}

First, we present a brief review of the properties of bundle maps over the identity on the cotangent bundle of a smooth manifold \citep{AMR88,Lee2012,marsden2013introduction}. 

Let $Q$  be a manifold and $\pi_Q: T^*Q\to Q$ the corresponding cotangent bundle.  A smooth map $\Psi:T^*Q\to T^*Q$ is said to be a bundle map over the identity if the following diagram is commutative 
\begin{center}
$$
\xymatrix{T^*Q \ar[rr]^{{\Psi}}\ar[dr]_{\pi_Q}&&T^*Q\ar[dl]^{\pi_Q}\\&Q&}
$$
\end{center}

These kinds of maps are related to the semi-basic $1$-forms $\Theta$  on $T^*Q.$ We recall that $\Theta\in \Omega^1(T^*Q)$ is semi-basic if and only if
$$\langle\Theta(\alpha), X_{\alpha}\rangle=0, \mbox{ for all $\alpha\in T^*Q$ and $X_{\alpha}\in V_{\alpha}\pi_Q.$}$$

Here $V\pi_Q$ is the vertical sub-bundle of $\pi_Q.$ 

\begin{proposition}
There is a one-to-one correspondence between bundle maps of the cotangent bundle $T^*Q$ over the identity and semi-basic $1$-forms on $T^*Q.$
\end{proposition}

\begin{proof}
The correspondence 
$$\Xi:\{\Psi:T^*Q\to T^*Q/\pi_Q\circ \Psi=\pi_Q\}\to
\{\Theta\in \Omega^1(T^*Q)/\mbox{ $\Theta$ is  semi-basic}\}$$ is given by 
$$\Xi(\Psi)=\Psi^*(\Theta_Q),$$
where $\Theta_Q\in \Omega^1(T^*Q)$ is the Liouville $1$-form on $T^*Q,$ i.e. 
$$\langle\Theta_Q(\alpha), X_{\alpha}\rangle=\langle\alpha, T_\alpha\pi_Q(X_{\alpha})\rangle,\mbox{ for all }\alpha\in T^*Q.$$

Note that 
$$\begin{array}{rcl}
\langle(\Psi^*(\Theta_Q))(\alpha), X_{\alpha}\rangle&=&\langle\Theta_Q(\Psi(\alpha)), T_{\alpha}\Psi(X_{\alpha})\rangle=\langle\Psi({\alpha}), (T_{\Psi(\alpha)}\pi_Q\circ T_{\alpha}\Psi)(X_{\alpha})\rangle =\langle\Psi({\alpha}), (T_{\alpha}\pi_Q)(X_{\alpha})\rangle.
\end{array}$$

Then, $\Xi(\Psi)$ is a semi-basic $1$-form on $T^*Q$.

On the other hand, the construction of the inverse of $\Xi$ 
is as follows: 
the semi-basic $1$-form  $\Theta\in \Omega^1(T^*Q)$ induces the bundle map $\Psi_{\Theta}:T^*Q\to T^*Q$ given by 

$$\langle\Psi_{\Theta}(\alpha),v\rangle=\langle\Theta(\alpha), Y_{\alpha}\rangle,$$
with $Y_\alpha\in T_\alpha(T^*Q)$ and $(T_{\alpha}\pi_Q)(Y_{\alpha})=v.$ 

This map is well-defined because if $T_{\alpha}\pi_Q(Z_{\alpha})=v,$ then
$Z_{\alpha}-Y_{\alpha}\in V_{\alpha}\pi_Q$. Therefore, 
$\langle\Theta(\alpha),Z_{\alpha}-Y_{\alpha}\rangle=0$ 
and 
$$\langle\Theta(\alpha), Y_{\alpha}\rangle=\langle\Theta(\alpha), Z_{\alpha}\rangle.$$
\end{proof}

\begin{remark}
If $({\bf q})$ are local  coordinates on $Q$, $({\bf q},{\bf p})$ are the corresponding local coordinates on $T^*Q$, and the bundle map is given locally by 
$$\Psi({\bf q},{\bf p})=(q^i,\psi_i({\bf q},{\bf p})),$$
then the associated semi-basic $1$-form is just 
$$(\Psi^*\Theta_Q)(q^i,p_i)=\psi_i({\bf q},{\bf p})dq^i.$$

Conversely, if $\Theta({\bf q},{\bf p})=\theta_i({\bf q},{\bf p})dq^i$ is the local expression of a semi-basic $1$-form on $T^*Q$, then the corresponding bundle map is 
$$\Psi({\bf q},{\bf p})=(q^i,\theta_i({\bf q},\bf p)).$$
\end{remark}

{\it A canonical transformation on $T^*Q$} is a diffeomorphism (not necessarily a bundle map over the identity) 
$\Psi:T^*Q\to T^*Q$ such that $\Psi^*\omega_Q=\omega_Q$, where $\omega_{Q}$ is the canonical symplectic form on $T^{*}Q$. More generally, {\it a canonoid transformation} for the Hamiltonian dynamics determined by the Hamiltonian vector field $X_{H}\in {\mathfrak X}(T^*Q)$ of the function $H:T^*Q\to {\Bbb R},$ is a diffeomorphism  (not necessarily a bundle map over the identity) $\Psi:T^*Q\to T^*Q$ for which there exists a function $K:T^*Q\to \mathbb{R}$ such that 
$$i_{X_{H}}\Psi^*(\omega_Q)=dK$$
(see \citep{CR88,CFR2013}). In this case, if $\omega_{\Psi}=\Psi^*(\omega_Q),$ we have that 
$${\mathcal L}_{X_H}\omega_\Psi=0.$$
There is a special type of canonoid transformation, those preserving the base manifold $Q$ are called fouling transformations. The  next notion is a generalization of this kind of transformation. 

\begin{definition}\label{defc}
{A fouling map for the Hamiltonian dynamics determined by the Hamiltonian vector field $X_{H}\in {\mathfrak X}(T^*Q)$ of the function $H:T^*Q\to {\Bbb R}$ (for short, a fouling map for $H\in C^{\infty}(T^{*}Q)$),} is a bundle map $\Psi:T^{*}Q\longrightarrow T^{*}Q$ over the identity for which the  $2$-form $\omega_\Psi=\Psi^*(\omega_Q)$ is invariant with respect to the Hamiltonian vector field $X_H$ of $H$,  i.e.,
\begin{equation}\label{compatible}
{\mathcal L}_{X_H}\omega_\Psi=d(i_{X_{H}}\omega_\Psi)=0.
\end{equation}
\end{definition} 

\begin{remark}
The notion of a fouling map is more general than that of a fouling transformation. Indeed, a fouling map is not necessarily a diffeomorphism and  it only requires that the 1-form $i_{X_H}\omega_\Psi$ is closed but not necessarily exact. In this case, the closed $2$-form $\omega_\Psi$ is not necessarily symplectic. In the case of fouling transformations, however, the 2-form $\omega_\Psi$ is symplectic.
\end{remark}

On $T^*Q,$  Hamiltonian vector fields of basic functions $f\circ \pi_Q:T^*Q\to {\Bbb R},$ with $f:Q\to \R,$ and of linear functions $Y^\ell:T^*Q\to\R$, with $Y$ a vector field  on $Q$, generate the space of the vector fields ${\mathfrak X}(T^*Q)$ on $T^*Q.$ Thus, 
    \begin{equation}\label{3}
    \Psi\mbox{ is a fouling map}\Leftrightarrow\left\{\begin{array}{rcl}
     d(i_{X_{H}}\omega_\Psi)(X_{f\circ \pi_Q}, X_{g\circ \pi_Q})&=&0\\[8pt]
    d(i_{X_{H}}\omega_\Psi)(X_{f\circ \pi_Q}, X_{Y^l})&=&0\\[8pt]
    d(i_{X_{H}}\omega_\Psi)(X_{Y^\ell}, X_{Z^\ell})&=&0
    \end{array}\right.
    \end{equation}
for all $f,g\in C^\infty(Q)$ and $Y,Z\in {\mathfrak X}(Q)$. 


Let $\Psi:T^*Q\to T^*Q$ be a bundle map over the identity. Then, we have a $(1,1)$-tensor field on $T^*Q$
$$L_\Psi:T^*(T^*Q)\to T^*(T^*Q)$$
defined by 
\begin{equation}
L_\Psi=\flat_{\omega_\Psi}\circ \#_{\omega_Q},
\end{equation}
where $\flat_{\omega_\Psi}:T(T^*Q)\to T^*(T^*Q)$ and $\#_{\omega_Q}:T^*(T^*Q)\to T(T^*Q)$ are the morphisms  of vector bundles characterized by  
$$\flat_{\omega_\Psi}(X_{\alpha})=i_{X_{\alpha}}(\omega_\Psi(\alpha)) \mbox{ and } i_{\#_{\omega_Q}(\alpha)}\omega_Q=\alpha.$$

\begin{remark}\label{3.4'}
Consider the dual morphism to $L_{\Psi},$ which we will denote by ${\mathcal N}$, that is,

\begin{equation}\label{3.3'}
    {\mathcal N}:=\#_{\omega_Q}\circ L_\Psi:T(T^*Q)\to T(T^*Q)
\end{equation}
and the $2$-form $\omega_{\mathcal N}$ on $T^*Q$ given by
\begin{equation}\label{3.3''}
    \omega_{\mathcal N}(\alpha)(X_\alpha,Y_\alpha)=\omega_Q(\alpha)({\mathcal N}_\alpha X_\alpha,Y_\alpha)\mbox{ for all }\alpha\in T^*Q\mbox{ and } X_\alpha,Y_\alpha\in T_\alpha(T^*Q).
\end{equation}
From (\ref{3.3'}) and (\ref{3.3''}), it is easy to prove that 
\begin{equation}\label{3.3'''}
\omega_{\mathcal N}=\omega_{\Psi}=\Psi^*(\omega_Q).
\end{equation}
So, the $2$-form $\omega_{\mathcal N}$ is closed. In addition, using (\ref{3.3''}), it follows that 
$$\omega_Q(\alpha)({\mathcal N}_\alpha X_\alpha,Y_\alpha)=\omega_\Psi(\alpha)(X_\alpha,Y_\alpha)=-\omega_{\Psi}(\alpha)(Y_\alpha,X_\alpha)=-\omega_Q(\alpha)({\mathcal N}_\alpha Y_\alpha,X_\alpha)=\omega_Q(\alpha)(X_\alpha,{\mathcal N}_\alpha Y_\alpha).$$

However, the $(1,1)$-tensor field ${\mathcal N}$ is not, in general, Nijenhuis. So, the couple $(\omega_Q,{\mathcal N})$ is not,  in general, a symplectic-Nijenhuis structure, following the terminology in \citep{MM84}. Therefore, although
$$i_{X_H}\omega_Q=dH, \;\;\; i_{X_H}\omega_{\mathcal N}=dK,$$
(at least, for a local function $K$ defined in an open neighborhood of every point), we do not obtain a bi-Hamiltonian system associated with a  presymplectic-Nijenhuis structure.

\end{remark}

Despite the observation stated in Remark \ref{3.4'}, we will see that the $(1,1)$-tensor field $L_\Psi$
determines constants of motion of the Hamiltonian dynamics induced by ${H}$. The following lemma will be necessary in the proof. 
\begin{lemma}\label{lema} 
    Let $T:T^*M\to T^*M$ be a $(1,1)$-tensor field on a manifold $M$ and $X$ a vector field on $M$. If ${\mathcal L}_X T=0$ then the trace $\mbox{\rm Tr }T$ of $T$ is an invariant of the flow of $X$, i.e., ${\mathcal L}_{X}(\mbox{\rm Tr }T)=0$.
\end{lemma}
\begin{proof}
    Let $\{e_i\}$ be a local basis of ${\mathfrak X}(M)$ and $\{e^i\}$ the corresponding dual basis of $\Omega^1(M)$. If $X=\displaystyle\sum_jX^je_j$ and we identify the tensor $T$ with a bilinear map $T:TM\times_M T^*M\to C^\infty(M),$ then 
$$\begin{array}{rcl}
    X(\mbox{\rm Tr }T)&=&X(T(e_i,e^i))=({\mathcal L}_XT)(e_i,e^i) + T([X,e_i],e^i) + T(e_i,{\mathcal L}_Xe^i)=T([X,e_i],e^i) + T(e_i,dX^i  + i_Xde^i)\\[8pt]&=&-e_i(X^j)T(e_j,e^i)+ X^jc_{ji}^kT(e_k,e^i)+ e_j(X^i)T(e_i,e^j) + X^kde^i(e_k,e_j)T(e_i,e^j)\\[8pt]&=&X^kc^i_{kj}T(e_i,e^j)-X^kc_{kj}^iT(e_i,e^j)=0,
\end{array}$$
where $[e_i,e_j]=c_{ij}^ke_k.$ 
\end{proof}

\begin{theorem}\label{th2.5}
Let $\Psi:T^*Q\to T^*Q$ be a fouling map for $H\in C^{\infty}(T^{*}Q)$. Then, for all $l\in \mathbb{N},$ the trace of $L_\Psi^l=L_{\Psi}\circ \cdots^{(l} \cdots\circ  L_{\Psi}$ is a constant of motion of $X_H.$
\end{theorem}
\begin{proof}
We will see that ${\mathcal L}_{X_H}L_\Psi=0,$ which, by Lemma \ref{lema}, implies that $X_H(\mbox{ Tr}L_\Psi)=0.$ Firstly, we will prove that 
\begin{equation}\label{Ls}{
\mathcal L}_{X_H}\circ \#_{\omega_Q}=\#_{\omega_Q}\circ {\mathcal L}_{X_H}.
\end{equation}
In fact, if $F\in C^\infty(T^*Q)$ then
\begin{equation}\label{cH}
[X_H,X_F]=-X_{\{H,F\}},
\end{equation}
where $\{\cdot,\cdot\}:C^\infty(T^*Q)\times C^\infty(T^*Q)\to C^\infty(T^*Q)$ is the Poisson bracket on $T^*Q$ (see, for instance, \citep{AMRC2008}). Therefore,  we have that 
$$({\mathcal L}_{X_H}\circ \#_{\omega_Q})(dF)=[X_H, X_F]=-X_{\{H,F\}}=\#_{\omega_Q}(d\{F,H\})=\#_{\omega_Q}((dX_H(F)))=(\#_{\omega_Q}\circ {\mathcal L}_{X_H})(dF).$$

Now, we show that
\begin{equation}\label{Lb}
{\mathcal L}_{X_H}\circ \flat_{\omega_\Psi}=\flat_{\omega_\Psi}\circ {\mathcal L}_{X_H}.
\end{equation}
Indeed, for all $F\in C^\infty(T^*Q)$, using (\ref{compatible}), we have 
$${\mathcal L}_{X_H}(\flat_{\omega_\Psi}(X_F))={\mathcal L}_{X_H}(i_{X_F}\omega_\Psi)=i_{X_F}{\mathcal L}_{X_H}\omega_\Psi + i_{[{X_H},{X_F}]}\omega_\Psi=\flat_{\omega_\Psi}({\mathcal L}_{X_H}X_F).$$
Since the Hamiltonian vector fields  $X_F$ generate the space of vector fields on $T^*Q,$ we deduce that (\ref{Lb}) holds. 

Using (\ref{Ls}) and (\ref{Lb}), we have that 
\begin{equation}\label{Invar-L-k}
    {\mathcal L}_{X_H}L_\Psi=0.
\end{equation}
In fact, if $F\in C^\infty(T^*Q)$
$$({\mathcal L}_{X_H}L_\Psi)(dF)={\mathcal L}_{X_H}(L_\Psi(dF))-L_\Psi({\mathcal L}_{X_H}(dF))={\mathcal L}_{X_H}(\flat_{\omega_\Psi}(\#_{\omega_Q}(dF)))-\flat_{\omega_\Psi}(\#_{\omega_Q}({\mathcal L}_{X_H}(dF)))=0.$$

Finally,  we will proceed by induction on $l.$ From (\ref{Invar-L-k}) and Lemma \ref{lema}, we deduce that the result holds for $l=1$. Now,  we suppose that it is true for $l-1$ and we will prove it for $l:$ from this hypothesis and (\ref{Invar-L-k})
$$({\mathcal L}_{X_H}L^l_\Psi)(dF)={\mathcal L}_{X_H}(L^{l-1}_\Psi(L_\Psi(dF)))-L^{l-1}_\Psi(L_\Psi({\mathcal L}_{X_H}(dF)))=({\mathcal L}_{X_H}\circ L^{l-1}_\Psi-L^{l-1}_\Psi\circ {\mathcal L}_{X_H})(L_\Psi(dF))=0.$$

So, using Lemma \ref{lema}, we conclude the theorem. 
\end{proof}

{\bf Local expressions.-} Suppose that 
$$\Psi(q,p)=(q^i,\theta_i(q,p))$$
Then, 
$$\omega_{\Psi}=\Psi^*\omega_Q=\frac{\partial \theta_i}{\partial q^j}dq^i\wedge dq^j+\frac{\partial \theta_i}{\partial p_j}dq^i\wedge dp_j $$
and 
$$L_{\Psi}(dq^j)=\frac{\partial \theta_i}{\partial p_j}dq^i,\;\;\;\; L_{\Psi}(dp_k)=(\frac{\partial \theta_k}{\partial q^i}-\frac{\partial{\theta_i}}{\partial q^k})dq^i+ \frac{\partial{\theta_k}}{\partial p_j}dp_j. $$
Then, it is easy to prove that for every $l\in \mathbb{N},$
$$\mbox{\rm Tr }L_\Psi^l=\langle L_\Psi^l(dq^j),\frac{\partial }{\partial q^j}\rangle + \langle L_\Psi^l(dp_k),\frac{\partial }{\partial p_k}\rangle=2 \frac{\partial \theta_{i_1}}{\partial p_{i_2}}\frac{\partial \theta_{i_2}}{\partial p_{i_3}}\dots \frac{\partial \theta_{i_{l-1}}}{\partial p_{i_l}}\frac{\partial \theta_{i_l}}{\partial p_{i_1}}.$$

So, for every $l\in \mathbb{N}$
\begin{equation}\label{8'}f_l(q,p)=2 \frac{\partial \theta_{i_1}}{\partial p_{i_2}}\frac{\partial \theta_{i_2}}{\partial p_{i_3}}\dots \frac{\partial \theta_{i_{l-1}}}{\partial p_{i_l}}\frac{\partial \theta_{i_l}}{\partial p_{i_1}}
\end{equation}
is a first integral of $X_H.$ Note that if the functions $\theta_i$ are homogeneous polynomials, for every $i$, then the first integrals $f_l$ are also homogeneous, for every $l.$

\section{A particular case: mechanical Hamiltonian functions}
\label{sec3}

In this section, we will present a description of  polynomial fouling maps for mechanical Hamiltonian functions, induced by symmetric tensors of type $(s+1,0)$, with $s\in \{0,\dots ,k\}$.

Let $(Q,g)$ be a semi-Riemannian manifold of dimension $n$.
\begin{definition}\citep{godbillon1969geometrie,AMRC2008}
A mechanical Hamiltonian function on $(Q,g)$ is a real smooth function $H:T^{*}Q\longrightarrow \mathbb{R}$ of the form
\begin{equation}
H=(g^{*})^{hom}+V\circ \pi_{Q},
\end{equation}
with $V\in C^{\infty}(Q)$, $g^{*}:T^*Q\times_QT^*Q\to {\Bbb R}$ the dual metric on $T^*Q$  induced by  $g$, and $(g^{*})^{hom}:T^{*}Q\longrightarrow \mathbb{R}$ the homogeneous function associated to $g^{*}$ given by $$(g^{*})^{hom}(\alpha_{q})=\frac{1}{2}g^{*}_{q}(\alpha_{q},\alpha_{q})\quad \forall \alpha_{q}\in T_{q}^{*}Q.$$ 
\end{definition}
The homogeneous function $(g^{*})^{hom}$ represents the kinetic energy in the space $T^*Q$, and $V\in C^\infty(Q)$ is the potential energy function; thus, a mechanical Hamiltonian function is the total energy of a physical system. Hamiltonian systems on $T^{*}Q$ with mechanical Hamiltonian functions are called natural Hamiltonian systems \citep{Arnold78}.

Let $(q^{i},p_{i})=(q^{1},\ldots,q^{n},p_{1},\ldots, p_{n})$ be local canonical coordinates on $T^{*}Q$, i.e., $\omega_Q=dq^{i}\wedge dp_{i}$. Suppose that the local expression of $g$ is 
$$g=g_{ij}dq^{i}\otimes dq^{j},$$
then,  the local expression of the dual metric $g^*$ 
is given by   
\begin{equation}
g^{*}=g^{ij}\partial_{q^{i}}\otimes\partial_{ q^{j}},
\end{equation}
where $(g^{ij})$ is the inverse matrix of $(g_{ij})$, and therefore,  
\begin{equation}
(g^{*})^{hom}(q^{i},p_{i})=\frac{1}{2}g^{ij}p_{i}p_{j}.
\end{equation}

Given a mechanical Hamiltonian function $H:T^{*}Q\longrightarrow \mathbb{R}$, the local expression of the Hamiltonian vector field for $H$ is 
\begin{equation}\label{hamiltonian}
X_H=g^{mn}p_m\partial_{q^n} + (\Gamma_{sm}^ng^{sl}p_np_l-\partial_{q^m}V)\partial_{p_m},
\end{equation}
where $\Gamma_{sm}^n$ are the Christoffel symbols of the metric $g$. 

We  are interested in homogeneous bundle maps over the identity on $T^*Q$. In what follows, we recall the relation of these maps and the symmetric tensors on $Q.$

For each symmetric $(k+1,0)$-tensor field $T$ on $Q$,  we have the bundle map $\Psi_{T}:T^{*}Q\longrightarrow T^{*}Q$ over the identity given by
\begin{equation}\label{Tk+1}
\Psi_{T}(\alpha_{q})(X_q)=\frac{1}{k!}\widetilde{T}(\alpha_q,\dots ,\alpha_q)(X_q),
\end{equation}
where $\widetilde{T}$ is the $(k,1)$-tensor induced by $T$, that is, 
\begin{equation}\label{Tt}
\widetilde{T}_{q}(\alpha_1,\ldots,\alpha_k)(v)=
T_q(\alpha_1,\ldots,\alpha_k,\beta) \mbox{ for all } \alpha_i\in T_{q}^{*}Q \mbox{ and } v\in T_qQ,
\end{equation}
where $\beta\in T_q^*Q$ is the covector $\beta_q=g_q(v,\cdot).$

If the local expression of $T$ is 
\begin{equation}
T=T^{i_{1}\ldots i_{k+1}}\partial_{q^{i_{1}}}\otimes\ldots\partial_{ q^{i_{k+1}}},
\end{equation}
with $T^{i_{1}\ldots i_{k+1}}$ local functions on $Q$, the corresponding $(k,1)$-tensor deduced from $T$ is 
$$\widetilde{T}=\widetilde{T}^{i_1\dots i_k}_{i}\partial_{q^{i_{1}}}\otimes\ldots\partial_{ q^{i_{k}}}\otimes dq^i$$
such that $$\widetilde{T}^{i_1\dots i_k}_i=g_{ij}T^{i_1\dots i_kj}.$$
Thus, the bundle map $\Psi_T:T^*Q\to T^*Q$ is given by 
\begin{equation}\label{lcPhi}
\Psi_{T}(q^{i},p_{i})=(q^{i},\frac{1}{k!}\widetilde{T}^{i_{1}\ldots i_{k}}_ip_{i_{1}}\cdots p_{i_{k}})=(q^{i},\frac{1}{k!}g_{ij}{T}^{i_{1}\ldots i_{k}j}p_{i_{1}}\cdots p_{i_{k}}).
\end{equation}

In addition, from (\ref{Liouville}) and (\ref{lcPhi}), it follows that $\Psi_T$ is a $k$-homogeneous smooth map, that is, 
$$\Psi_T(\lambda\alpha_q)=\lambda^k\Psi_T(\alpha_q) \mbox{ for } \lambda\in \mathbb{R} \mbox{ and } \alpha_q\in T^*_qQ$$
and for every $q\in Q$, the map 
$$T_q^*Q\times \cdots ^{k} \cdots \times T_q^*Q\to T_q^*Q$$
defined by
\begin{equation}\label{m2}
\begin{array}{rcl}
(\alpha^1_q,\dots ,\alpha_q^k) &\longrightarrow& \displaystyle\frac{1}{k!} \left((-1)^{k+1}\displaystyle\sum^k_{i=1} \Psi_T(\alpha_q^i)+(-1)^{k}\displaystyle{\sum^k_{i<j}} \Psi_T(\alpha^i_q + \alpha^j_q)+ \cdots \right. \\
&& \left. \dots -\displaystyle{\sum^k_{i}} \Psi_T(\alpha^1_q +\dots+ \alpha^i_q+\alpha^{i+1}_q + \cdots +  \alpha_k)+ \Psi_T(\alpha^1_q  + \dots +  \alpha_q^k)\right)
\end{array}
\end{equation}
is multilinear. In fact, (\ref{m2}) defines a symmetric tensor of type $(k,1)$ on $Q$ that is just $\widetilde{T}.$

In consequence, there is a one-to-one correspondence between homogeneous bundle maps of degree $k$ on $T^*Q$ over  the identity and symmetric $(k+1,0)$-tensors on $Q$. 

Now, we will study the case when the fouling map $\Psi:T^*Q\to T^*Q$ for a mechanical Hamiltonian function is 
induced by a symmetric tensor $T_k$ of type $(k+1,0)$.

\subsection{The case when the tensor is a vector field $T_0$ on $Q$}
In this subsection, we will characterize fouling maps associated with a $(1,0)$-tensor field $T_0$ on $Q$ (that is, a vector field $T_0$ on $Q$). The corresponding bundle map $\Psi_{T_0}:T^*Q\to T^*Q$ is a $0$-homogeneous map. 

In what follows, we denote by 
$\sharp:T^*Q\to TQ$
the isomorphism of vector bundles induced by $g^*$, that is, 
\begin{equation}\label{sharp}
\langle \beta, \sharp \alpha\rangle=g^*(\alpha,\beta)\mbox{ for all } \alpha,\beta\in \Omega^1(Q)
\end{equation}
and by $\flat:TQ\to T^*Q$ the inverse isomorphism $\sharp^{-1},$ that is 
$$\flat(X)(Y)=g(X,Y),\;\;\; X,Y\in {\mathfrak X}(Q).$$
Note that 
$$\sharp (dq^i)=g^{ij}\partial_{q^j} \mbox{ and } \flat(\partial_{q^i})=g_{ij}dq^j.$$
Both isomorphisms of vector bundles, $\sharp$ and $\flat$,  induce  $C^\infty(Q)$-module morphisms 
$$\sharp:\Omega^1(Q)\to {\mathfrak X}(Q),\;\;\; \flat:{\mathfrak X}(Q)\to \Omega^1(Q).$$

\begin{lemma}\label{nablasos}
For all $X\in {\mathfrak X}(Q)$ and $\gamma\in \Omega^1(Q).$ 
\begin{enumerate} 
\item[(a)]
$\nabla_X\sharp\gamma=\sharp(\nabla_X\gamma),$
\item[(b)] $\nabla_X(\flat(Y))=\flat(\nabla_X Y).$
\end{enumerate}
\end{lemma}
\begin{proof}
(a) Since the metric tensor is parallel, 
$$\begin{array}{rcl}
0=(\nabla_Xg)(\sharp\gamma,Z)&=&X(g(\sharp\gamma,Z))-g(\nabla_X\sharp\gamma,Z)- g(\sharp\gamma,\nabla_XZ)\\[8pt]&&=X(\langle\gamma,Z\rangle)-g(\nabla_X\sharp\gamma,Z)-\langle\gamma,\nabla_XZ\rangle\\[8pt]&&=\langle\nabla_X\gamma,Z\rangle-g(\nabla_X\sharp\gamma,Z)=g(\sharp(\nabla_X\gamma),Z)-g(\nabla_X\sharp\gamma,Z)\end{array}$$
for all $Z\in {\mathfrak X}(Q).$

(b) Let $\alpha\in \Omega^1(Q)$ be the $1$-form such that $Y=\sharp \alpha$. Then, from (a), 
$$\nabla_XY=\sharp (\nabla_X\flat(Y)), $$
which is equivalent to (b). 
\end{proof}

Let $R_{T_0}$ be the $2$-vector on $Q$ defined by 
\begin{equation}\label{T0}
R_{T_0}(\alpha,\beta)= (\nabla_{\sharp\alpha}{T_0})(\beta)-(\nabla_{\sharp\beta}{T_0})(\alpha), \mbox{ with }\alpha,\beta\in \Omega^1(Q),
\end{equation}
where $\nabla$ is the Levi-Civita connection associated with the metric $g.$ 

We will characterize the fact that $\Psi_{T_0}:T^*Q\to T^*Q$ is a fouling map for a mechanical Hamiltonian function $H:T^*Q\to \mathbb{R}$ in terms of this $2$-vector. Previously, we prove the following lemma, which describes the $2$-vector $R_{T_0}$ using  the $1$-form   metrically equivalent to $T_0$ via the metric $g$, i.e, the $1$-form $\flat(T_0).$
\begin{lemma}
\label{R-T}
If $T_0$ is a vector field on the semi-Riemannian manifold $(Q,g)$, then
\begin{equation}
R_{T_0}(\alpha,\beta)=d(\flat(T_0))(\sharp\alpha,\sharp\beta),
\end{equation}
for all $\alpha,\beta\in \Omega^1(Q).$ 
\end{lemma}
\begin{proof}
Using  that the Levi-Civita connection is symmetric, it follows that 
$[\sharp\alpha,\sharp\beta]=\nabla_{\sharp\alpha}\sharp\beta-\nabla_{\sharp\beta}\sharp\alpha.$ So, we have 
\begin{equation}\label{26'-}
\begin{array}{rcl}
d(\flat(T_0))(\sharp\alpha,\sharp\beta)&=&\sharp\alpha(\langle \beta,T_0\rangle)-\sharp\beta(\langle\alpha,T_0\rangle)-g([\sharp\alpha,\sharp\beta],T_0)\\[8pt]&=&
\sharp\alpha(\langle \beta,T_0\rangle)-g(T_0,\nabla_{\sharp\alpha}\sharp\beta)-\sharp\beta(\langle \alpha,T_0\rangle)+g(T_0,(\nabla_{\sharp\beta}\sharp\alpha)),
\end{array}
\end{equation}
where $\langle\cdot,\cdot\rangle$ is the natural pairing between $\Omega^1(Q)$ and ${\mathfrak X}(Q).$

So, from  (\ref{T0}), (\ref{26'-}) and Lemma \ref{nablasos},  we conclude that 
$$d(\flat(T_0))(\sharp\alpha,\sharp\beta)=R_{T_0}(\alpha,\beta).$$
\end{proof}

Now, we  have
\begin{proposition}\label{propT0}
    Let $T_0$ be a vector field on $Q$  and  $H:T^*Q\to \mathbb{R}$ be a mechanical Hamiltonian function. Then, 
    \begin{enumerate}
    \item[(i)]
     $d(i_{X_{H}}\omega_{T_0})(X_{f\circ \pi_Q}, X_{h\circ \pi_Q})=0$ 
    \item[(ii)] $d(i_{X_{H}}\omega_{T_0})(X_{f\circ \pi_Q}, X_{Y^l})=(i_{\flat(Y)}R_{T_0})(f)\circ \pi_Q$
    \item[(iii)] $d(i_{X_{H}}\omega_{T_0})(X_{Y^\ell}, X_{Z^\ell})=({\mathcal L}_Yi_{\flat(Z)}R_{T_0}-{\mathcal L}_Zi_{\flat(Y)}R_{T_0}-i_{\flat({[Y,Z])}}R_{T_0})^{hom},$
    \end{enumerate}
with $f,h\in C^\infty(Q)$ and $Y,Z\in {\mathfrak X}(Q),$ where $\omega_{T_0}=\Psi_{T_0}^*(\omega_Q)$.
\end{proposition}
\begin{proof}
First,  we will prove that 
\begin{equation}\label{0}
\omega_{T_0}(X_H,X_{f\circ \pi_Q})=0.
\end{equation}
Let $(q^i)$ be coordinates on $Q$ and $(q^i,p_i)$ the corresponding coordinates  on $T^*Q.$  Then,  the local expressions of $T_0$ and $\Psi_{T_0}$ are  $$T_0=T^i\partial_{q^i}\mbox{ and } \Psi_{T_0}({\bf q},{\bf p})=(q^i,\widetilde{T}_i({\bf q}))=(q^i,g_{ij}{T}^i({\bf q})),$$ with $T^i$ local functions on $Q$, and $\widetilde{T}_i$  the local functions associated with  $1$-form $\flat({T}_0)=\widetilde{T}_idq^i$ metrically equivalent to $T_0$.
     
Moreover, locally 
\begin{equation}\label{4.15'}\omega_{T_0}=\partial_{q^k}\widetilde{T}_idq^i\wedge dq^k,\end{equation}
which implies that 
\begin{equation}\label{wT0}
\omega_{T_0}=-\Pi_Q^*(d(\flat({T}_0)).
\end{equation}

Now, using (\ref{hamiltonian}) and (\ref{4.15'}), it follows that
\begin{equation}\label{1}
i_{X_H}\omega_{T_0}=(\partial_{q^i}\widetilde{T}_k-\partial_{q^k}\widetilde{T}_i)g^{km} p_mdq^i.
\end{equation}
Therefore, since $X_{f\circ \pi_Q}$ is vertical,  we deduce (\ref{0}).

\noindent $(i)$ Using  (\ref{cH}) and (\ref{0}), we have that 
$$\begin{array}{rcl}
d(i_{X_{H}}\omega_{T_0})(X_{f\circ \pi_Q}, X_{h\circ \pi_Q})&=&X_{f\circ \pi_Q}(\langle i_{X_{H}}\omega_{T_0},X_{h\circ \pi_Q}\rangle)-X_{h\circ \pi_Q}(\langle i_{X_{H}}\omega_{T_0},X_{f\circ \pi_Q}\rangle)-\langle i_{X_{H}}\omega_{T_0},[X_{f\circ \pi_Q},X_{h\circ \pi_Q}]\rangle\\[8pt]&=&\langle i_{X_{H}}\omega_{T_0},X_{\{f\circ \pi_Q,h\circ \pi_Q\}}\rangle.
\end{array}$$

Finally, since 
\begin{equation}\label{b1}
\{f\circ \pi_Q,g\circ \pi_Q\}=0,
\end{equation}
we deduce that 
$$d(i_{X_{H}}\omega_{T_0})(X_{f\circ \pi_Q}, X_{h\circ \pi_Q})=0.$$

\noindent $(ii)$ If the local expression of $Y\in{\mathfrak X}(Q)$ is $Y=Y^k\partial_{ q^k}\in {\mathfrak X}(Q),$ then the linear function $Y^\ell\in C^\infty(T^*Q)$ is locally $Y^\ell=Y^k p_k$ and  we have that 
\begin{equation}\label{24'}
X_{Y^\ell}=Y^k\partial_{q^k}-p_k\partial_{ q^j}Y^k\partial_{p_j}.
\end{equation}
Using this and  (\ref{1}), we deduce that 
$$\omega_{T_0}(X_H,X_{Y^\ell})=(\partial_{q^i}\widetilde T_k-\partial_{q^k}\widetilde T_i)g^{mk}Y^ip_m.$$
On the other hand, we have 
$$d(\flat(T_0))(Y,\cdot )=(\partial_{q^i}\widetilde{T}_k-\partial_{q^k}\widetilde{T}_i)Y^idq^k.$$
In particular,
$$d(\flat(T_0))(Y,\sharp dq^m)=(\partial_{q^i}\widetilde{T}_k-\partial_{q^k}\widetilde{T}_i)Y^ig^{mk}.$$

Therefore, using Lemma \ref{R-T}, we conclude  that 
\begin{equation}\label{4}
\omega_{T_0}(X_H,X_{Y^\ell})=(i_{\flat(Y)}R_{T_0})^{hom}.
\end{equation}
Then, using   (\ref{cH}), (\ref{0}) and (\ref{4}), it follows that 
$$\begin{array}{rcl}
d(i_{X_H}\omega_{T_0})(X_{f\circ \pi_Q}, X_{Y^\ell}) &=&X_{f\circ \pi_Q}(\langle i_{X_H}\omega_{T_0},X_{Y^\ell}\rangle)-X_{Y^\ell}(\langle i_{X_H}\omega_{T_0},X_{f\circ \pi_Q}\rangle)-\langle i_{X_H}\omega_{T_0},[X_{f\circ \pi_Q},X_{Y^\ell}]\rangle\\[8pt]&=&X_{f\circ \pi_Q}(\langle i_{X_H}\omega_{T_0},X_{Y^\ell}\rangle)+\langle i_{X_H}\omega_{T_0},X_{Y(f)\circ \pi_Q}\rangle=X_{f\circ \pi_Q}((i_{\flat(Y)}R_{T_0})^{hom})\\[8pt]&=&({i_{\flat(Y)}R_{T_0}})(f)\circ \pi_Q.
\end{array}$$
Here, we have also used 
\begin{equation}\label{b2}
\{f\circ \pi_Q,Y^\ell\}=Y(f)\circ \pi_Q.
\end{equation}
So, we have $(ii).$ 

$(iii)$ Since 
\begin{equation}\label{b3}
\{Y^\ell,Z^\ell\}=-[Y,Z]^\ell,
\end{equation}
from again (\ref{cH}) and (\ref{4}), 
$$\begin{array}{rcl}
d(i_{X_H}\omega_{T_0})(X_{Y^\ell}, X_{Z^\ell})&=&X_{Y^\ell}(\langle i_{X_H}\omega_{T_0},X_{Z^\ell}\rangle)-X_{Z^\ell}(\langle i_{X_H}\omega_{T_0},X_{Y^\ell}\rangle)-\langle i_{X_H}\omega_{T_0},[X_{Y^\ell},X_{Z^\ell}]\rangle\\[8pt]&=&X_{Y^\ell}(\langle i_{X_H}\omega_{T_0},X_{Z^\ell}\rangle)-X_{Z^\ell}(\langle i_{X_H}\omega_{T_0},X_{Y^\ell}\rangle)-\langle i_{X_H}\omega_{T_0},X_{[Y,Z]^\ell}\rangle\\[8pt]&=&X_{Y^\ell}((i_{\flat(Z)}R_{T_0})^{hom})-X_{Z^\ell}((i_{\flat(Y)}R_{T_0})^{hom})-(i_{\flat([Y,Z])}R_{T_0})^{hom}.\end{array}$$

Now, using item (b) of Proposition \ref{lemma1}, we conclude that 
$$d(i_{X_H}\omega_{T_0})(X_{Y^\ell}, X_{Z^\ell})=({\mathcal L}_Yi_{\flat(Z)}R_{T_0}-{\mathcal L}_Zi_{\flat(Y)}R_{T_0} -i_{\flat({[Y,Z]})}R_{T_0})^{hom}.$$
\end{proof}

An immediate consequence of this result and Lemma \ref{R-T} is the following characterization 
\begin{corollary}
Let $T_0$ be a vector field  on $Q$ and  $H:T^*Q\to \mathbb{R}$ be a mechanical Hamiltonian function on $(Q,g)$. Then, the associated bundle map over the identity $\Psi_{T_0}:T^*Q\to T^*Q$ is a fouling map for $H$ if and only if the $2$-vector $R_{T_0}$ given in (\ref{T0}) is zero; equivalently, the $1$-form metrically equivalent to $T_0$, via the metric $g$ is closed.
\end{corollary}

Note that if  the $1$-form $\flat(T_0)$ is closed, then, using (\ref{wT0}), we have that the $2$-form $\omega_{T_0}$ is trivial, which implies that the tensor field $L_{\Psi_{T_0}}$ of type $(1,1)$ on $T^*Q$ is also trivial. So, for every $k\in \mathbb{N},$ the trace of $L_{\Psi_{T_0}}^k$ is the zero function on $T^*Q$ and, via Theorem \ref{th2.5}, we don't obtain non-trivial first integrals of the mechanical Hamiltonian vector field $X_H$.

\subsection{The case when the tensor  is a $(k+1,0)$-tensor $T_k$ on $Q$}
In this subsection, we extend the previous  results for fouling maps associated with $(k+1,0)$-tensor fields for $k\geq 1$. In this case, the corresponding bundle maps are $k$-homogeneous maps.

\begin{proposition}\label{WT}
Let $T_k$ be a symmetric $(k+1,0)$-tensor field on $Q,$ with $k\geq  1$. If   $H:T^*Q\to \mathbb{R}$ is a mechanical Hamiltonian function, with kinetic energy induced by the semi-Riemannian metric $g$ and  with potential function $V:Q\to \mathbb{R}$, then for all $f\in C^\infty(Q)$ and $Y\in {\mathfrak X}(Q),$
\begin{enumerate}
\item[(a)] $\omega_{T_k}(X_H, X_{f\circ \pi_Q})=k(i_{df}{T_k})^{hom}=kX_{f\circ \pi_Q}({T_k}^{hom}),$
\item[(b)]
$\omega_{T_k}(X_H,X_{Y^l})= (k{\mathcal L}_YT_k -i_{\flat(Y)}R_{T_k})^{hom} +(i_{dV}i_{\flat(Y)}{T_k})^{hom},$
\end{enumerate}
where $\omega_{T_k}=\Psi_{T_k}^*(\omega_Q)$, $\Psi_{T_k}:T^*Q\to T^*Q$ is the bundle map over the identity of $Q$ induced by $T_k$ and $R_{T_k}$ is the $(k+2,0)$-tensor given by 
\begin{equation}\label{R}
R_{T_k}(\alpha_0,\alpha_1,\dots,\alpha_{k+1})=\sum_{i=1}^{k+1} (\nabla_{\sharp\alpha_i}{T_k})(\alpha_1,\dots,\alpha_{i-1} ,\alpha_{0},\alpha_{i+1},\dots, \alpha_{k+1})-(\nabla_{\sharp\alpha_{0}}{T_k})(\alpha_1, \dots ,\alpha_{k+1}).
\end{equation}

Here $\nabla$ denotes the Levi-Civita connection of $(Q,g).$ 
\end{proposition}

For a proof of this result, see Appendix \ref{Appendix}. 

\begin{remark}
From the properties of the Levi-Civita connection, one deduces that $R_{T_k}$ is a $(k+2,0)$-tensor that, in general, is not symmetric. However,  for each $\alpha\in \Omega^1(Q),$ the  $(k+1,0)$-tensor  $i_{\alpha}R_{T_k}$ is  symmetric. In fact, if $X\in {\mathfrak X}(Q)$ and $T$ is a symmetric $(r,0)$-tensor then we have that $\nabla_XT$ is also a symmetic $(r,0)$-tensor. So, if $i,j\in\{1,\dots ,k+1\}$ and $i<j$, it follows that 
$$\begin{array}{rcl}
R_{T_k}(\alpha,\alpha_1,\dots,\alpha_i,\dots, \alpha_j,\dots,\alpha_{k+1})&=& 
(\nabla_{\sharp\alpha_1}{T_k})(\alpha,\alpha_2,\dots,\alpha_i,\dots, \alpha_j,\dots,\alpha_{k+1}) +\dots \\[8pt]&&
+(\nabla_{\sharp\alpha_i}{T_k})(\alpha,\alpha_1,\dots,\alpha_{i-1},\alpha_{i+1},\dots, \alpha_j,\dots,\alpha_{k+1})
+\dots \\[8pt]&&
+(\nabla_{\sharp\alpha_j}{T_k})(\alpha,\alpha_1,\dots,\alpha_{i},\dots, \alpha_{j-1},\alpha_{j+1},\dots,\alpha_{k+1})+\dots \\[8pt]&&
+(\nabla_{\sharp\alpha_{k+1}}{T_k})(\alpha,\alpha_1,\dots,\alpha_{i},\dots, \alpha_j,\dots,\alpha_{k})\\[8pt]&&
-(\nabla_{\sharp\alpha}{T_k})(\alpha_1,\dots,\alpha_{i},\dots, \alpha_j,\dots,\alpha_{k+1})\\[8pt]&=&R_{T_k}(\alpha,\alpha_1,\dots,\alpha_j,\dots, \alpha_i,\dots,\alpha_{k+1}).
\end{array}$$
Thus, $i_\alpha R_{T_k}$ is symmetric and the expressions $(i_{\flat(Y)}R_{T_k})^{hom}$ and $(i_{dV}i_{\flat(Y)}R_{T_k})^{hom}$ in the second item (b) of Proposition \ref{WT} make sense according our definition (\ref{2.0}) of the homogeneous function of degree $k$ associated with a symmetric $(k,0)$-tensor. 
\end{remark}

\begin{proposition}\label{PropTk}
Let $T_k$ be a symmetric $(k+1,0)$-tensor  on $Q$ with $k\geq 1$. If   $H:T^*Q\to \mathbb{R}$ is a mechanical Hamiltonian function on $(Q,g)$ with kinetic energy induced by a semi-Riemannian metric $g$ on $Q$ and with potential $V:Q\to \mathbb{R}$, then 
$$\begin{array}{llll}
(i)&
d(i_{X_{H}}\omega_{T_k})(X_{f\circ \pi_Q}, X_{h\circ \pi_Q})&=&0,\\[8pt]
(ii)& d(i_{X_{H}}\omega_{T_k})(X_{f\circ \pi_Q}, X_{Y^l})&=&-(i_{df}i_{\flat(Y)}R_{T_k})^{hom}+ (i_{dV}i_{df}i_{\flat(Y)}{T_k})^{hom},\\[8pt]
(iii)&
d(i_{X_{H}}\omega_{T_k})(X_{Y^\ell}, X_{Z^\ell})&=&-\left({\mathcal L}_Yi_{\flat(Z)}R_{T_k}-{\mathcal L}_Zi_{\flat(Y)}R_{T_k}-i_{\flat([Y,Z])}R_{T_k}\right)^{hom}\\[8pt]&&&+\left({\mathcal L}_Y i_{dV}i_{\flat(Z)}{T_k}-{\mathcal L}_Z i_{dV}i_{\flat(Y)}{T_k} \right. -\left. i_{dV}i_{\flat_{[Y,Z]}}T_k\right)^{hom}, 
\end{array}$$
where $\omega_{T_k}=\Psi_{T_k}^*\omega_Q,$ with $f,h\in C^\infty(Q)$ and $Y,Z\in {\mathfrak X}(Q)$.
\end{proposition}
\begin{proof}
$(i)$ From (\ref{cH}), (\ref{b1}) and item $(a)$ in Proposition \ref{WT}, we have that
$$\begin{array}{rcl}d(i_{X_{H}}\omega_{T_k})(X_{f\circ \pi_Q}, X_{h\circ \pi_Q})&=& X_{f\circ \pi_Q}(i_{X_{H}}\omega_{T_k}(X_{h\circ \pi_Q}))-X_{h\circ \pi_Q}(i_{X_{H}}\omega_{T_k}(X_{f\circ \pi_Q}))-i_{X_{H}}\omega_{T_k}([X_{f\circ \pi_Q},X_{h\circ \pi_Q}])\\[8pt]&=&
kX_{f\circ \pi_Q}(X_{h\circ \pi_Q}({T_k}^{hom})) - kX_{h\circ \pi_Q}(X_{f\circ \pi_Q}({T_k}^{hom}))+i_{X_{H}}\omega_{T_k}(X_{\{f\circ \pi_Q,h\circ \pi_Q\}})\\[8pt]&=&k[X_{f\circ\pi_Q}, X_{h\circ \pi_Q}]({T_k}^{hom})=-kX_{\{f\circ \pi_Q,h\circ \pi_Q\}}({T_k}^{hom})=0.
\end{array}$$
Thus,  $(i)$ is proved.

$(ii)$ Using (\ref{b2}), Propositions \ref{lemma1}  and \ref{WT}, 
$$\begin{array}{rcl}
d(i_{X_{H}}\omega_{T_k})(X_{f\circ \pi_Q}, X_{Y^\ell})&=& X_{f\circ \pi_Q}(i_{X_{H}}\omega_{T_k}(X_{Y^\ell}))-X_{Y^\ell}(i_{X_{H}}\omega_{T_k}(X_{f\circ \pi_Q}))-i_{X_{H}}\omega_{T_k}([X_{f\circ \pi_Q},X_{Y^\ell}])\\[8pt]&=&
X_{f\circ \pi_Q}\left(\left( k({\mathcal L}_YT_k) -i_{\flat(Y)}R_{T_k}\right)^{hom}  + (i_{dV}i_{\flat(Y)}{T_k})^{hom}\right) - kX_{Y^\ell}(X_{f\circ \pi_Q}({T_k}^{hom}))\\[8pt]&&+i_{X_{H}}\omega_{T_k}(X_{\{f\circ \pi_Q,Y^\ell\}})
\\[8pt]&=&(ki_{df}{\mathcal L}_Y{T_k}-i_{df}i_{\flat(Y)}R_{T_k})^{hom}+ (i_{dV}i_{df}i_{\flat(Y)}{T_k})^{hom}-k({\mathcal L}_Yi_{df}T_k)^{hom}\\[8pt]&&+ i_{X_{H}}\omega_{T_k}(X_{Y(f)\circ \pi_Q})
\\[8pt]&=&-k(i_{d(Y(f))}T_k)^{hom}
-(i_{df}i_{\flat (Y)}R_{T_k})^{hom}+(i_{dV}i_{df}i_{\flat(Y)}{T_k})^{hom}+ k(i_{d(Y(f))}T_k)^{hom}\\[8pt]&=&
-(i_{df}i_{\flat(Y)}R_{T_k})^{hom}+ (i_{dV}i_{df}i_{\flat(Y)}{T_k})^{hom}.
\end{array}$$
This proves $(ii)$.

$(iii)$ If $X,Y\in {\mathfrak X}(Q)$ then, using 
Proposition \ref{WT}, item (b) in Proposition \ref{lemma1}, (\ref{cH}) and (\ref{b3}), we deduce that
$$\begin{array}{rcl}
d(i_{X_{H}}\omega_{T_k})(X_{Y^\ell}, X_{Z^\ell})&=& X_{Y^\ell}(i_{X_{H}}\omega_{T_k}(X_{Z^\ell}))-X_{Z^\ell}(i_{X_{H}}\omega_{T_k}(X_{Y^\ell}))-i_{X_{H}}\omega_{T_k}([X_{Y^\ell},X_{Z^\ell}])\\[8pt]&=&X_{Y^\ell}\left((k{\mathcal L}_ZT_k-i_{\flat(Z)}R_{T_k})^{hom}+ i_{dV}i_{\flat(Z)}{T_k})^{hom}\right)\\[8pt]
&&-X_{Z^\ell}\left((k{\mathcal L}_YT_k-i_{\flat(Y)}{R_{T_k}})^{hom}+(i_{dV}i_{\flat(Y)}{T_k})^{hom}\right) + i_{X_{H}}\omega_{T_k}(X_{\{Y^\ell,Z^\ell\}})\\[8pt]&=&-({\mathcal L}_Yi_{\flat(Z)}{R_{T_k}})^{hom}+(k{\mathcal L}_Y{\mathcal L}_ZT_k)^{hom}+({\mathcal L}_Y i_{dV}i_{\flat(Z)}{T_k})^{hom}+({\mathcal L}_Zi_{\flat(Y)}{R_{T_k}})^{hom} \\[8pt]&&-k({\mathcal L}_Z{\mathcal L}_YT_k)^{hom}-({\mathcal L}_Z i_{dV}i_{\flat(Y)}{T_k})^{hom} -i_{X_{H}}\omega_{T_k}(X_{[Y,Z]^\ell})
\\[8pt]&=&-({\mathcal L}_Yi_{\flat(Z)}{R_{T_k}})^{hom}+k({\mathcal L}_Y{\mathcal L}_ZT_k)^{hom}+({\mathcal L}_Y i_{dV}i_{\flat(Z)}{T_k})^{hom}+({\mathcal L}_Zi_{\flat(Y)}{R_{T_k}})^{hom}\\[8pt]&&-k({\mathcal L}_Z{\mathcal L}_YT_k)^{hom}-({\mathcal L}_Z i_{dV}i_{\flat(Y)}{T_k})^{hom} + (i_{\flat([Y,Z])}R_{T_k})^{hom} - k({\mathcal L}_{[Y,Z]}{T_k})^{hom} \\[8pt]&&- (i_{dV}i_{\flat([Y,Z])}{T_k})^{hom}\\[8pt]&=&
-\left({\mathcal L}_Yi_{\flat(Z)}{R_{T_k}} - {\mathcal L}_Zi_{\flat(Y)}{R_{T_k}}-i_{\flat({[Y,Z]})}R_{T_k}\right)^{hom}+\\[8pt]&&\left
({\mathcal L}_Y i_{dV}i_{\flat(Z)}{T_k}-{\mathcal L}_Z i_{dV}i_{\flat(Y)}{T_k} -i_{dV}i_{\flat([Y,Z])}{T_k}\right)^{hom}
\end{array}$$
\end{proof}

Now, we obtain the following result.
\begin{theorem}\label{the}
Let $T_k$ be a symmetric $(k+1,0)$-tensor field on $Q,$ with $k\geq 1$.  If  $H:T^*Q\to \mathbb{R}$ is a mechanical Hamiltonian function with kinetic energy induced by a semi-Riemannian metric $g$ on $Q$ and a potential energy $V:Q\to \mathbb{R}$, then 
\begin{enumerate}
\item for $k>1$, $\Psi_{T_k}:T^{*}Q\longrightarrow T^{*}Q$  is a fouling map for $H$ if and only if $R_{T_k}=0$ and 
 $i_{dV}{T_k}=0,$
\item for  $k=1$, $\Psi_{T_1}:T^{*}Q\longrightarrow T^{*}Q$  is a fouling map for $H$ if and only if $R_{T_1}=0$ and $d(\flat(i_{dV}T_1))=0,$
\end{enumerate}
where $R_{T_{k}}$ is the $(k+2,0)$-tensor defined in (\ref{R}).
\end{theorem}
\begin{proof}
From (\ref{3}), and Proposition \ref{PropTk}, $\Psi_{T_k}$ is a fouling map for $H$ if and only if 
\begin{enumerate} 
\item[(a)] $-(i_{df}i_{\flat(Y)}R_{T_k})^{hom}+ (i_{dV}i_{df}i_{\flat(Y)}{T_k})^{hom}=0$
\item[(b)] 
$ -\left({\mathcal L}_Yi_{\flat(Z)}{R_{T_k}} - {\mathcal L}_Zi_{\flat(Y)}{R_{T_k}}-i_{\flat({[Y,Z]})}R_{T_k}\right)^{hom}+\left(
{\mathcal L}_Y i_{dV}i_{\flat(Z)}{T_k}-{\mathcal L}_Z i_{dV}i_{\flat(Y)}{T_k} -i_{dV}i_{\flat([Y,Z])}{T_k}\right)^{hom}=0$
\end{enumerate}
for all $Y,Z\in {\mathfrak X}(M)$ and $f\in C^\infty(Q).$ 

In the case $k=1$ (a) is equivalent to $R_{T_1}=0.$ 
In such a case, (b) holds if and only if $d(\flat(i_{dV}{T_1}))=0.$ In fact, 
$$0={\mathcal L}_Y i_{\flat(Z)}{i_{dV}T_1}-{\mathcal L}_Z i_{\flat(Y)}{i_{dV}T_1} -i_{\flat({[Y,Z])}}i_{dV}T_1$$
is equivalent to 
$$0=Y(\flat(i_{dV}T_1)(Z))-Z(\flat(i_{dV}T_1)(Y)) -(\flat(i_{dV}T_1))([Y,Z])=d(\flat(i_{dV}T_1))(Y,Z),$$
for all $Y,Z\in {\mathfrak X}(Q).$

If $k>1$, since $(i_{df}i_{\flat(Y)}R_{T_k})^{hom}$ is  a homogeneous map of degree $k$ and $(i_{df}i_{\flat(Y)}i_{dV}{T_k})^{hom}$ is one of degree $k-2$, then (a) is equivalent to $$(i_{df}i_{\flat(Y)}R_{T_k})^{hom}=0\mbox{ and } (i_{df}i_{\flat(Y)}i_{dV}{T_k})^{hom}=0.$$ 
Now, using that $i_{df}i_{\flat(Y)}R_{T_k}$ and $i_{df}i_{\flat(Y)}i_{dV}{T_k}$ are symmetric tensors, then (a) is equivalent to 
$$R_{T_k}=0 \mbox{ and } i_{dV}{T_k}=0.$$ 

Moreover, if $(a)$ is satisfied, then (b) holds. 
\end{proof}

\noindent {\bf Local expressions: } Suppose that 
$$T_k=T^{i_1\dots i_ki_{k+1}}\partial_{q^{i_1}}\otimes \dots \otimes \partial_{q^{i_{k+1}}}$$
is a symmetric tensor field of type $(k+1,0)$ such that $R_{T_k}=0$ and $i_{dV}T_k=0$ if $k>1$ or $R_{T_1}=0$ and $d(\flat(i_{dV}T_1))=0$ if $k=1.$ So, the corresponding bundle map $\Psi_{T_k}:T^*Q\to T^*Q$ over the identity of $Q$ is given by 
$$\Psi_{T_k}(q^i,p_i)=(q^i,\displaystyle\frac{1}{k!}g_{ij}T^{i_1\dots i_kj}p_{i_1}\dots p_{i_k}).$$

Moreover, if we denote by 
$$\theta_i(q,p)=\frac{1}{k!}g_{ij}T^{i_1\dots i_kj}p_{i_1}\dots p_{i_k}, \forall i,$$
then, using  (\ref{8'}), we deduce that, for every $l\in \mathbb{N},$ 
\begin{equation}\label{flk}
\begin{array}{rcl}
f_{lk}(q,p)&=&
2g_{i_1j_1}T^{i_2j_2^1\dots j_k^1j_1}g_{i_2j_2}T^{i_3j_2^2\dots j_k^2j_2}\dots
g_{i_{l-1}j}T^{i_lj_2^{l-1}\dots j_k^{l-1}j_{l-1}}g_{{i_l}j_{l}}T^{i_1j_2^l\dots j_k^lj_{l}}
p_{j_2^1}\dots p_{j_k^1}\\&& p_{j_2^2}\dots p_{j_k^2}\dots p_{j_2^{l-1}}\dots p_{j_k^{l-1}}p_{j_2^{l}}\dots p_{j_k^{l}}
\end{array}
\end{equation}
is a $l(k-1)$-homogeneous first integral of $X_H.$ 

From Theorem \ref{the}, we deduce the following result. 

\begin{corollary}
\label{4.8'}
Let $T_k$ be a symmetric $(k+1,0)$-tensor field on $Q$, with $k\geq 1$. If $H:T^*Q\to {\Bbb R}$ is the kinetic energy induced by a semi-Riemannian metric $g$ on $Q$ then $\Psi_{T_k}:T^*Q\to T^*Q$ is a fouling map for $H$ if and only if $R_{T_k}=0.$
\end{corollary}

\begin{remark}
If $H$ is the kinetic energy induced by a semi-Riemannian metric $g$ on $Q$ and $R_{T_k}=0$ then, from (\ref{flk}), we have that 
$$\mbox{\rm Tr}(L^l_{\Psi_{T_k}})$$
is a homogeneous first integral of $X_H,$ for every $l\in {\Bbb N}.$ This implies that the symmetric $(k+1,0)$-tensor $T_k$ induces a hierarchy of Killing tensors on the semi-Riemannian manifold $(Q,g)$ (see Section \ref{secmotivation}).
\end{remark}

\subsection{The general case when $\Psi$ is a polynomial map}
Now, for mechanical Hamiltonian functions,  we consider polynomial bundle maps $\Psi:T^*Q\to T^*Q$ over the identity of order $k$, i.e., $\Psi=\Psi_{T_{0}}+\Psi_{T_{1}}+\cdots+\Psi_{T_{k}}$ as a fiberwise sum, such that for each $s\in\{0,\ldots,k\},$ $T_{s}$ is a symmetric $(s+1,0)$-tensor field on $Q$, and $\Psi_{T_s}$ is the corresponding $s$-homogeneous  bundle map over the identity induced by $T_s$ (see (\ref{Tk+1}), (\ref{Tt}) and (\ref{lcPhi})). 

If the local expression of $T_s$ is
$$T_s=T_s^{i_{1}\ldots i_{s+1}}\partial_{q^{i_{1}}}\otimes\ldots\otimes\partial_{ q^{i_{s+1}}},$$
then 
\begin{equation}\label{Psi-local-exp}
\Psi(q_i,p_i)=(q_i,g_{ij}({\bf q})(T_0^{j}({\bf q}) + \sum_{s=1,\dots k}\frac{1}{s!} T_s^{i_1\dots, i_sj}({\bf q})p_{i_1}\dots p_{i_s})).
\end{equation}

Moreover, we will prove the following theorem. 
\begin{theorem}
\label{theoremPsi}
Let $T_{0}$ be a vector field on $Q$ and $T_{k}$ be a symmetric $(k+1,0)$-tensor field on $Q$ for each $k\in\{1,2,\ldots,s\}$. If  $H:T^*Q\to \mathbb{R}$ is a mechanical Hamiltonian function with kinetic energy induced by a semi-Riemannian metric $g$ on $Q$ and a potential energy $V:Q\to \mathbb{R}$, then 
$\Psi_T=\displaystyle\sum_{k=0}^{s}\Psi_{T_k}:T^{*}Q\longrightarrow T^{*}Q$  is a fouling map for $H$ if and only if 
\begin{enumerate}
 \item[(i)] $R_{{T}_k}=i_{dV}{T}_{k+2}$ for $k=0,\dots s-2,$
 \item[(ii)] 
$R_{{T}_k}=0$  for $k=s,s-1,$
\item[(iii)]
$d(\flat({i_{dV}T_1}))=0.$
\end{enumerate}
where $R_{T_k}$  is the $(k+2,0)$-tensor on $Q$ defined as in (\ref{R}) for the symmetric $(k+1,0)$-tensor $T_k.$ 
\end{theorem}
\begin{proof}
If $T=\displaystyle\sum_{k=0}^sT_k$, then using (\ref{Psi-local-exp}), it is easy to see that
$$\omega_T=\Psi_T^*\omega_Q=\sum_{k=0}^s \omega_{T_k}.$$
So, this implies that 
\begin{equation}\label{4.23'}
d(i_{X_H}\omega_T)=\sum_{k=0}^s d(i_{X_H}\omega_{T_k}).\end{equation}

Thus, from $(i)$ in  Propositions \ref{propT0} and \ref{PropTk}, we have that
$$d(i_{X_{H}}\omega_{T})(X_{f\circ \pi_Q}, X_{h\circ \pi_Q})=0$$
On the other hand, from (\ref{4.23'}) and (ii) in Propositions \ref{propT0} and \ref{PropTk},

\begin{equation}\label{iip}
\begin{array}{rcl}
d(i_{X_{H}}\omega_{T})(X_{f\circ \pi_Q}, X_{Y^l})&=&-\displaystyle\sum_{k=0}^s(i_{df}i_{\flat(Y)}R_{T_k})^{hom}+ \displaystyle\sum_{k=2}^s(i_{dV}i_{df}i_{\flat(Y)}T_k)^{hom},\end{array}
\end{equation}
for all $f\in C^\infty(Q)$ and $Y\in {\mathfrak X}(Q)$. In the same way, using (iii) in these propositions, we deduce that 
\begin{equation}\label{iiip}
\begin{array}{rcl}
d(i_{X_{H}}\omega_{T})(X_{Y^\ell}, X_{Z^\ell})&=&-\displaystyle\sum_{k=0}^s\left({\mathcal L}_Yi_{\flat(Z)}R_{T_k}-{\mathcal L}_Zi_{\flat(Y)}R_{T_k}-i_{\flat([Y,Z])}R_{T_k}\right)^{hom}\\[8pt]&&+\displaystyle\sum_{k=1}^s\left({\mathcal L}_Y i_{dV}i_{\flat(Z)}{T_k}-{\mathcal L}_Z i_{dV}i_{\flat(Y)}{T_k} -i_{dV}i_{\flat_{[Y,Z]}}T_k\right)^{hom}.\end{array}
\end{equation}
Then $\Psi_T=\displaystyle\sum_{k=0}^{s}\Psi_{T_k}:T^{*}Q\longrightarrow T^{*}Q$  is a fouling map for $H$ if and only if (\ref{iip}) and (\ref{iiip}) are null. Now, (\ref{iip}) being zero is equivalent to 
$$-\sum_{k=0}^s(i_{df}i_{\flat(Y)}R_{T_k})^{hom} + \sum_{k=2}^s(i_{df}i_{\flat(Y)}i_{dV}T_k)^{hom}=0.$$
So, equalizing homogeneous functions of the same degree, we have that 
$$(i_{df}i_{\flat(Y)}R_{T_{k}})^{hom}=(i_{df}i_{\flat(Y)}i_{dV}T_{k+2})^{hom} \mbox{ for all } k=0,\dots, s-2,$$$$(i_{df}i_{\flat(Y)}R_{T_{s-1}})^{hom}=0\mbox{ and }(i_{df}i_{\flat(Y)}R_{T_s})^{hom}=0.$$
Note that $i_{df}i_{\flat(Y)}R_{T_k}$ and $i_{df}i_{\flat(Y)}i_{dV}T_{k+2}$  are symmetric tensors. Thus, 
\begin{equation}\label{Tk}
R_{T_s}=0,\;\; R_{T_{s-1}}=0\mbox{ and }R_{T_k}=i_{dV}T_{k+2}, \mbox{ for all } k=0,\dots s-2.
\end{equation}

Finally, if we suppose that (\ref{Tk}) holds, then the nullification of (\ref{iiip}) is equivalent to
$$0={\mathcal L}_Y i_{\flat(Z)}{i_{dV}T_1}-{\mathcal L}_Z i_{\flat(Y)}{i_{dV}T_1} -i_{\flat_{[Y,Z]}}i_{dV}T_1.$$

But, proceeding as in the proof of item (1) in Theorem \ref{the}, we conclude that this condition is equivalent to the fact that  the $1$-form $\flat({i_{dV}T_1})$ on $Q$ to be closed. 
\end{proof}

\begin{remark}
\begin{enumerate}
\item[(a)]
Since $R_{T_0}$ is a $2$-vector and $T_2$ is symmetric, the condition $R_{T_0}=i_{dV}T_2$ is equivalent to $R_{T_0}=0$ and $i_{dV}T_2=0.$

\item[(b)] Let us remember that given a fouling map $\Psi_T$ for a mechanical Hamiltonian function $H$, there exists an associated $(1,1)-$tensor field $L_{\Psi}$ that is invariant under the Hamiltonian flow; therefore, the traces of its powers are constants of motion. In the special case when $\Psi$ is a polynomial map $\Psi_T=\displaystyle\sum_{k=0}^{s}\Psi_{T_k}:T^{*}Q\longrightarrow T^{*}Q$, such constants of motion are polynomial functions of the momentum coordinates. In fact, 
$$L_{\Psi_T}=\sum_{k=0}^sL_{\Psi_{T_k}}$$
which implies that 
$$L_{\Psi_{T}}^l=\sum_{i_1,\dots, i_l=0}^s L_{\Psi_{T_{i_1}}}\circ \dots \circ L_{\Psi_{T_{i_l}}}\mbox{ for }l\in {\Bbb N}$$
and, therefore, 
$$\mbox{\rm Tr}(L^l_{\Psi_{T}})=\sum_{i_1,\dots i_l=0}^s \mbox{\rm Tr}(L_{\Psi_{T_{i_1}}}\circ\dots \circ  L_{\Psi_{T_{i_l}}}).$$
In particular, 
$$\mbox{\rm Tr}(L_{\Psi_{T}})=\sum_{k=0}^s\mbox{\rm Tr}(L_{\Psi_{T_k}}).$$
\end{enumerate}
\end{remark}
From Corollary \ref{4.8'} and Theorem \ref{theoremPsi}, we deduce the following result.
\begin{corollary}
Let  $T_{k}$ be a symmetric $(k+1,0)$-tensor field on $Q,$ for each $k\in\{0,\ldots,s\}$. If  $H:T^*Q\to \mathbb{R}$ is the kinetic energy induced by a semi-Riemannian metric $g$ on $Q$,  then the following conditions are equivalent:
\begin{enumerate}
\item[(i)]$\Psi_T=\displaystyle\sum_{k=0}^{s}\Psi_{T_k}:T^{*}Q\longrightarrow T^{*}Q$  is a fouling map for $H.$
\item[(ii)] For every $k\in \{0,\dots,s\}, R_{T_k}=0.$
\item[(iii)] For every $k\in \{0,\dots,s\},$ $\Psi_{T_k}:T^*Q\to T^*Q$ is a fouling map for $H$. 
\end{enumerate}
\end{corollary}

\section{Fouling maps for mechanical Hamiltonian functions on the Euclidean plane}
\label{sec4}

Let us consider $Q=E^{2}$ with the plane metric \citep{lee2018introduction,carroll2004introduction} given in cartesian coordinates $(x,y)$ by $g=dx^{2}+dy^{2}$, or in matrix form by
\begin{equation*}
(g_{ij})=\left( \begin{array}{cc}
1&0\\
0&1
\end{array} \right)=(g^{ij}).
\end{equation*}
The expression in coordinates $(x,y)$ of mechanical Hamiltonian functions on $(E^{2},g)$ is of the form
\begin{equation}
H=\frac{1}{2}(p_{x}^{2}+p_{y}^{2})+V(x,y),
\end{equation}
for $V\in C^{\infty}(E^{2})$.
Here $(x,y,p_x,p_y)$ are coordinates on $T^*E^2.$

\subsection{The case of polynomial maps of degree 1}
Let $T_{0}$ and $T_{1}$ be a vector field  and a symmetric $(2,0)$-tensor field on $E^{2},$ respectively.  From Theorem \ref{theoremPsi}, we deduce that $\Psi_T=\Psi_{T_0} + \Psi_{T_1}$ is a fouling map for $H$ if and only if 
\begin{enumerate}
    \item[(i)] $R_{T_0}=0$,
    \item[(ii)] $R_{{T}_{1}}=0$ ,
    \item[(iii)] $d(\flat(i_{dV}{T}_{1}))=0.$
\end{enumerate}

In coordinates $(x,y)$ we write
\begin{equation*}
T_{0}=A_{0}\partial_x+B_{0}\partial_y,\quad A_{0},B_{0}\in C^{\infty}(E^{2}).
\end{equation*}
The condition (i) is equivalent to $d(\flat(T_0))=d(A_0dx+B_0dy)=0$, that is, 
\begin{equation}\label{condT0}
\partial_xB_{0}-\partial_yA_{0}=0.
\end{equation}

In what follows,   $Sym(S)$ denotes  the symmetrization  of the tensor $S$. 
So, for the symmetric tensor $T_1$, we have that 
\begin{equation}\label{T1}
{T}_{1}=A_1\partial_x\otimes \partial_x+B_1 Sym(\partial_x\otimes \partial_y)+C_1\partial_y\otimes \partial_y,
\end{equation}
with $A_1,B_1,C_1\in C^{\infty}(E^{2}).$

Condition (ii) can be read
\begin{equation}\label{RT1}
\left\lbrace \begin{array}{l}
R_{{T}_1}(dx,dx,dx)={\partial_x}A_{1}=0,\;\;\;\;
R_{{T}_1}(dy,dx,dx)=2{\partial_x}B_{1}-{\partial_y}A_{1}=0,\\[8pt]
R_{{T}_1}(dx,dy,dx)={\partial_y}A_{1}=0,\;\;\;\;
R_{{T}_1}(dx,dy,dy)=2{\partial_y}B_{1}-{\partial_x}C_{1}=0,\\[8pt]
R_{{T}_1}(dy,dx,dy)={\partial_x}C_{1}=0,\;\;\;\;
R_{{T}_1}(dy,dy,dy)={\partial_y}C_{1}=0\ .
\end{array} \right.
\end{equation}
The solution to this system of differential equations is $A_{1}=a_{1},B_{1}=b_{1},C_{1}=c_{1}\in\mathbb{R}.$

Finally, 
\[
\flat(i_{dV}T_1)=
(a_1\partial_xV+b_1\partial_yV)\,dx
+
(b_1\partial_xV+c_1\partial_yV)\,dy.
\]
So, this $1$-form is closed (condition (iii)) if and only if  
\begin{equation}\label{condT1}(a_1-c_1)\partial_y\partial_xV+b_1(\partial_y\partial_yV-\partial_x\partial_xV)=0.
\end{equation}

In conclusion, $\Psi_T=\Psi_{T_0} + \Psi_{T_1}$ is a fouling map for $H$ if and only if $T_{0}=A_{0}\partial_x+B_{0}\partial_y$ satisfies (\ref{condT0}),  the coefficients of $T_1$ are constant and (\ref{condT1}) holds.

\begin{example}
We consider the Hamiltonian function $H=\frac{1}{2}(p_{x}^{2}+p_{y}^{2})+\frac{1}{2}(x^{2}+y^{2})$ for the isotropic harmonic oscillator. 

Let us take $A_{0}=xy,B_{0}=\frac{1}{2}x^{2}$, then
\begin{equation}
{\partial_x}B_{0}-{\partial_y}A_{0}=x-x=0.
\end{equation}

 On the other hand, we have that 
(\ref{condT1}) holds  for arbitrary $a_{1},b_{1},c_{1}\in {\Bbb R}$. So, 
\begin{equation}
\Psi_T(x,y,p_{x},p_{y})=(x,y,xy+a_{1}p_{x}+b_{1}p_{y},\frac{1}{2}x^{2}+c_{1}p_{y}+b_{1}p_{x})
\end{equation}
is a fouling map for $H$. In fact, one has
\begin{equation}
i_{X_{H}}\omega_{\Psi_T}=d\left(\frac{a_1}{2}(p_x^{2}+x^{2})+\frac{c_1}{2}(p_y^{2}+y^{2})+b_1(p_xp_y+xy)\right),
\end{equation}
where the resulting three-parameter family $\lbrace \frac{a_1}{2}(p_x^{2}+x^{2})+\frac{c_1}{2}(p_y^{2}+y^{2})+b_1(p_xp_y+xy)\rbrace$ is just the space of quadratic integrals of the isotropic oscillator.

In the particular case $a_{1}=2,b_{1}=-1,c_{1}=1,$ we have the canonoid (fouling) transformation given in \citep{NOT87,CR88}. The associated constants of motion are trivial real constants. 
\end{example}

\subsection{The case of polynomial maps of degree 2}
Now,  let us consider a symmetric $(i+1,0)$-tensor field ${T}_i$ for each $i\in \lbrace 1,2\rbrace$, and a vector field $T_{0}$ on $E^2.$  The conditions in Theorem \ref{theoremPsi} reduce to:
\begin{enumerate}
    \item[(i)] $R_{{T}_0}=i_{dV}{T}_2,$
    \item[(ii)] $R_{T_{1}}=0$ and $R_{T_{2}}=0,$
    \item[(iii)]  $d(\flat(i_{dV}T_1))=0.$
\end{enumerate}

Note that $R_{{T}_0}$ is skew-symmetric and 
$i_{dV}{T}_2$ is symmetric. Then $R_{{T}_0}=0$ and $i_{dV}{T}_2=0,$ or equivalently, $d(\flat(T_0))=0$ and $i_{dV}{T}_2=0.$ 

From the symmetry of ${T}_{2},$ we have that 
\begin{equation}
\label{T2}
\begin{array}{rcl}
{T}_{2}=A_{2}\displaystyle\partial_x\otimes \partial_x\otimes\partial_x +B_{2}Sym(\partial_y\otimes\partial_x\otimes \partial_x)
+C_{2}Sym(\displaystyle\partial_x\otimes\partial_y\otimes \partial_y)+D_{2}\partial_y\otimes \partial_y\otimes\partial_y\end{array}
\end{equation}
with $A_{2},B_{2},C_{2},D_{2}\in C^{\infty}(E^{2})$.
Now
\begin{equation}\label{RT2}
\left\{
\begin{array}{l}
R_{{T}_{2}}(dx,dx,dx,dx)=2{\partial_x}A_2,\;\;\;\; R_{{T}_{2}}(dy,dx,dx,dx)=3{\partial_x}B_2-{\partial_y}A_2,\\[8pt]
R_{{T}_{2}}(dx,dx,dx,dy)={\partial_x}B_2+{\partial_y}A_2,\;\;\;\;
R_{{T}_{2}}(dy,dx,dx,dy)=2{\partial_x}C_2,\\[8pt] R_{{T}_{2}}(dx,dx,dy,dy)=2{\partial_y}B_2,\;\;\;\;
R_{{T}_{2}}(dy,dy,dy,dx)={\partial_y} C_2+{\partial_x}D_2,\\[8pt]
R_{{T}_{2}}(dx,dy,dy,dy)=3{\partial_y}C_2-{\partial_x}D_2,\;\;\;\;
R_{{T}_{2}}(dy,dy,dy,dy)=2{\partial_y}D_2.
\end{array}\right.
\end{equation}
Then, the solution of $R_{{T}_2}=0$ is $A_{2}=a_{2},B_{2}=b_{2},C_{2}=c_{2},D_{2}=d_{2}\in\mathbb{R}$.

Moreover,  the equation $i_{dV}{T}_2=0$ is equivalent to  
\begin{equation}
\label{systemT2}
\left\lbrace \begin{array}{c}
a_{2}\displaystyle{\partial_x}V+b_{2}{\partial_y}V=0,\\[8pt] 
b_{2}{\partial_x}V+c_{2}{\partial_y}V=0,\\[8pt]
c_{2}{\partial_x}V+d_{2}{\partial_y}V=0.
\end{array} \right.
\end{equation}

The conditions $R_{{T}_1}=0$ and $d(\flat(i_{dV}T_1))=0$ can be analyzed as in the polynomial case of order $1$. So 
$$T_{1}=a_{1}\partial_x\otimes \partial_x + b_{1}Sym(\partial_x\otimes {\partial_y})+ c_1{\partial_y}\otimes {\partial_y}$$ with $a_{1},b_{1},c_1\in \R$ satisfying (\ref{condT1}).

In conclusion,  $\Psi_T=\Psi_{T_{0}}+\Psi_{T_{1}}+\Psi_{T_{2}}$ is a fouling map for the mechanical Hamiltonian function $H=\frac{1}{2}(p_{x}^{2}+p_{y}^{2})+V(x,y)$ if and only if the $1$-form $\flat(T_0)$ is closed, the coefficients of the tensor $T_1$ are constant and satisfy (\ref{condT1}) and the coefficients of the tensor $T_2$ are constant and satisfy (\ref{systemT2}).
 
\begin{example}\label{5.2}
Let again $V(x,y)=\frac{1}{2}(x^{2}+y^{2})$. System (\ref{systemT2}) reduces to
\begin{equation}
\left\lbrace \begin{array}{c}
a_{2}x+b_{2}y=0,\\ 
b_{2}x+c_{2}y=0,\\
c_{2}x+d_{2}y=0,
\end{array} \right.
\end{equation}
which implies that $a_{2}=b_{2}=c_{2}=d_{2}=0$, therefore $T_{2}=0$. 
\end{example}

\begin{example}
Now let us consider the function $V(x,y)=x$. System (\ref{systemT2}) reduces to
\begin{equation}
\left\lbrace \begin{array}{c}
a_{2}=0,\\ 
b_{2}=0,\\
c_{2}=0,
\end{array} \right.
\end{equation}
Equation (\ref{condT1}) is satisfied for arbitrary $a_{1},b_{1},c_{1}$. So, for arbitrary $a_{1},b_{1},c_{1},d_{2}$, the map
\begin{equation}
\Psi_T(x,y,p_{x},p_{y})=(x,y,A_{0}+a_{1}p_{x}+b_{1}p_{y},B_{0}+c_{1}p_{y}+b_{1}p_{x}+\frac{1}{2}d_{2}p_{y}^{2})
\end{equation}
is a fouling map for $H=\frac{1}{2}(p_{x}^{2}+p_{y}^{2})+x$, with $A_0,B_0\in C^\infty(E^2)$ satisfying (\ref{condT0}).

By taking $A_{0}=B_0=0$ we have that the associated constants of motion (up to the power $2$ of the tensor $L_{\Psi_T}$) are respectively
\begin{equation}
f_{1}=2a_1+2c_1+2d_2p_{y}\quad\textit{and}\quad f_{2}=2a_1^2+4b_1^2+2(c_1+d_2p_y)^2.
\end{equation}
\end{example}

\subsection{The case of polynomial maps of degree 3}
Consider a symmetric $(i+1,0)$-tensor field $T_i$ for each $i\in \lbrace 1,2,3\rbrace$, and a vector field  $T_{0}$ on $E^2.$ The conditions in Theorem \ref{theoremPsi} reduce to
\begin{enumerate}
    \item[(i)] $R_{{T}_0}=i_{dV}T_{2}$ or equivalently, $R_{{ T}_0}=0$ and $i_{dV} T_{2}=0,$
    \item[(ii)] $d(\flat(i_{dV}T_{1}))=0$,
    \item[(iii)] $R_{ T_{1}}=i_{dV}{T}_{3}$, 
    \item[(iv)] $R_{T_{2}}=0$ and $R_{T_{3}}=0$.
\end{enumerate}

We suppose that  the symmetric tensor  ${T}_{3},$ is 
\begin{equation}
\begin{array}{rcl}
{T}_{3}&=&A_{3}\displaystyle\partial_x\otimes\partial_x\otimes\partial_x\otimes\partial_x 
+B_{3}Sym(\displaystyle\partial_x\otimes\partial_x\otimes\partial_x\otimes\partial_y)+C_{3}Sym
(\displaystyle\partial_x\otimes\partial_x\otimes\partial_y\otimes\partial_y)\\[8pt]&&
+D_{3}Sym(\displaystyle\partial_y\otimes\partial_y\otimes\partial_y\otimes\partial_x) + F_{3}\displaystyle\partial_y\otimes\partial_y\otimes\partial_y\otimes\partial_y
\end{array}
\end{equation}
with $A_{3},B_{3},C_{3},D_{3}, F_3\in C^{\infty}(E^{2})$.

Now, 
\begin{equation}
\left\{
\begin{array}{ll}
R_{T_3}(dx,dx,dx,dx,dx)=3\partial_xA_3,
 & R_{T_3}(dy,dx,dx,dx,dx)=4\partial_xB_3-\partial_yA_3,\\[8pt]
R_{T_3}(dx,dx,dx,dx,dy)=2\partial_xB_3+\partial_yA_3,
 & R_{T_3}(dy,dx,dx,dx,dy)=3\partial_xC_3,\\[8pt]
R_{T_3}(dx,dx,dx,dy,dy)=\partial_xC_3+2\partial_yB_3,
 & R_{T_3}(dy,dx,dx,dy,dy)=2\partial_xD_3+\partial_yC_3,\\[8pt]
R_{T_3}(dx,dx,dy,dy,dy)=3\partial_yC_3,
 & R_{T_3}(dy,dx,dy,dy,dy)=\partial_xF_3+2\partial_yD_3,\\[8pt]
R_{T_3}(dx,dy,dy,dy,dy)=4\partial_yD_3-\partial_xF_3,
 & R_{T_3}(dy,dy,dy,dy,dy)=3\partial_yF_3.
\end{array}
\right.
\end{equation}
The solution to the system $R_{T_3}=0$ is $A_{3}=a_{3},B_{3}=b_{3},C_{3}=c_{3},D_{3}=d_{3},F_{3}=f_{3}\in\mathbb{R}$.

On the other hand,
$$
\left\{
\begin{array}{l}
i_{dV}{T}_{3}(dx,dx,dx)=a_3{\partial_x}V + b_3{\partial_y}V,\;\;\;\; i_{dV}{T}_{3}(dx,dx,dy)=b_3{\partial_x}V + c_3{\partial_y}V,\\[8pt]

i_{dV}{T}_{3}(dx,dy,dx)=b_3{\partial_x}V + c_3{\partial_y}V,\;\;\;\; i_{dV}{T}_{3}(dy,dy,dx)=c_3{\partial_x}V + d_3{\partial_y}V,\\[8pt]

i_{dV}{T}_{3}(dx,dy,dy)=c_3{\partial_x}V + d_3{\partial_y}V,\;\;\;\;
i_{dV}{T}_{3}(dy,dy,dy)=d_3{\partial_x}V + f_3{\partial_y}V.
\end{array}\right.$$
So, ${T}_1$ is of the form given in (\ref{T1}), the equation $R_{T_1}=i_{dV}{T}_{3}$ (see (\ref{RT1})) is equivalent to
\begin{equation}\label{T1T3}
\left\{
\begin{array}{l}
{\partial_x}A_1=a_3
{\partial_x}V + b_3{\partial_y}V,\;\;\;\;  2{\partial_x} B_1-{\partial_y} A_1= b_3{\partial_x}V + c_3{\partial_y}V,\\[8pt]

{\partial_y}A_1=b_3{\partial_x}V + c_3{\partial_y}V,\;\;\;\; 2{\partial_y}B_1 - {\partial_x}C_1=c_3{\partial_x}V + d_3{\partial_y}V,\\[8pt]

{\partial_x }C_1=c_3{\partial_x}V + d_3{\partial_y}V,\;\;\;\;
{\partial_y}C_1=d_3{\partial_x}V + f_3{\partial_y}V.
\end{array}\right.
\end{equation}

Combining this equations we deduce that 
$${\partial_y}A_1={\partial_x}B_1, \mbox{ and } {\partial_x}C_1={\partial_y}B_1.$$
The condition $d(\flat(i_{dV}T_{1}))=0$ can be rewritten as

\begin{equation}\label{dVT1}
(C_{1}-A_{1}){\partial_x\partial_y}V+B_{1}(\partial_x\partial_xV-\partial_y\partial_yV)=0.
\end{equation}

Finally, $T_2$ is a symmetric tensor of type $(2,0)$ as (\ref{T2}) with $R_{{T}_2}=0$ and $i_{dV} T_2=0$. That is, its coefficients are constant and satisfy (\ref{systemT2}).

In conclusion,  $\Psi_T=\Psi_{T_{0}}+\Psi_{T_{1}}+\Psi_{T_{2}}+\Psi_{T_{3}}$ is a fouling map for the mechanical Hamiltonian function $H=\frac{1}{2}(p_{x}^{2}+p_{y}^{2})+V(x,y)$ if and only if the $1$-form $\flat(T_0)$ is closed, the coefficients of the tensor $T_1$ satisfy (\ref{dVT1}) , the coefficients of the tensors $T_2$ are constant satisfying (\ref{systemT2}),  $T_3$ are constant, and the coefficients of $T_3$ and those of $T_1$ are related by  (\ref{T1T3}) .
\begin{example}
Let once again $V(x,y)=\frac{1}{2}(x^{2}+y^{2})$. The coefficients of $T_{1}$ and $T_{3}$ are related by
 $$A_{1}=\displaystyle\frac{1}{2}a_3x^{2}+b_{3}xy+\frac{1}{2}c_{3}y^{2},$$ $$B_{1}=\displaystyle\frac{1}{2}b_{3}x^{2}+c_{3}xy+\frac{1}{2}d_{3}y^{2},$$ $$C_{1}=\displaystyle\frac{1}{2}c_{3}x^{2}+d_{3}xy+\frac{1}{2}f_{3}y^{2}.$$

Equation (\ref{dVT1}) is satisfied for arbitrary $a_{3},b_{3},c_{3},d_{3},f_{3}$.

In the particular case of $A_{0}=B_{0}=0$, $a_{3}=f_{3}=1$ and $b_{3}=c_{3}=d_{3}=0$, we have the bundle map
\begin{equation}
\Psi_T(x,y,p_{x},p_{y})=(x,y,\frac{1}{2}x^{2}p_{x}+\frac{1}{6}p_{x}^{3},\frac{1}{2}y^{2}p_{y}+\frac{1}{6}p_{y}^{3}),
\end{equation}
which is a canonoid transformation given in \citep{NOT87,CR88}.

The associated constants of motion (up to the power $2$ of the tensor $L_{\Psi_T}$) are
\begin{equation}
f_{1}=x^2 + y^2 + p_x^2 +p_y^2=2H\quad and \quad f_{2}=\frac{1}{2}(p_y^2 + y^2)^2 + \frac{1}{2}(p_x^2 + x^2)^2.
\end{equation}
Of course, the eigenvalues of the $(1,1)-$tensor field $L_{\Psi_T}$ are also constants of motion; in this case, they are also polynomial functions on the momentum coordinates, namely
\begin{equation}
e_{1}=\frac{1}{2}(p_x^2 + x^2)\quad\textit{and}\quad e_{2}=\frac{1}{2}(p_y^2 + y^2).
\end{equation}
\end{example}

\section{Fouling maps for mechanical Hamiltonian functions on the 2-sphere}
\label{sec5}

Let us consider $Q=S^{2}$ with the  metric \citep{lee2018introduction,carroll2004introduction} given in spherical coordinates $(\theta,\phi)$ by $g=d\theta^{2}+\sin^{2}\theta d\phi^{2}$, or in matrix form by
\begin{equation}
(g_{ij})=\left( \begin{array}{cc}
1&0\\
0&\sin^{2}\theta
\end{array} \right).
\end{equation}
We have that
\begin{equation}
(g^{ij})=\left( \begin{array}{cc}
1&0\\
0&\displaystyle\frac{1}{\sin^{2}\theta}
\end{array} \right).
\end{equation}

Then, 
$$\sharp(d\theta)=\partial_\theta \mbox{ and } \sharp(d\phi)=\frac{1}{\sin^2\theta}\partial_\phi$$
and therefore 
$$\flat(\partial_\theta)=d\theta \mbox{ and } \flat(\partial_\phi)=\sin^2\theta d\phi.$$

The non-null Christoffel symbols of this metric $g$ are 
\begin{equation}\label{sC}
\Gamma_{\phi\phi}^\theta=-\sin\theta\cos\theta\mbox{ and }\Gamma_{\theta\phi}^\phi=\Gamma_{\phi\theta}^\phi=\cot\theta.
\end{equation}

Moreover, 
\begin{equation}\label{nablaS2}
\nabla_{\partial_\theta}d\theta=0\;\;\; \nabla_{\partial_\theta}d\phi=-\cot\theta d\phi,\;\;\; \nabla_{\partial_\phi}d\theta=\sin\theta\cos\theta d\phi,\;\;\; \nabla_{\partial_\phi}d\phi=-\cot\theta d\theta.
\end{equation}

The expression in coordinates $(\theta,\phi)$ of mechanical Hamiltonian functions on $(S^{2},g)$ is of the form
\begin{equation}
H=\frac{1}{2}(p_{\theta}^{2}+\frac{1}{\sin^{2}\theta}p_{\phi}^{2})+V(\theta,\phi),
\end{equation}
for $V\in C^{\infty}(S^{2})$.

\subsection{The case of polynomial maps of degree 1}
Let $T_{0}$ and $T_{1}$ be a vector field and a symmetric $(2,0)$-tensor field on $S^{2},$ respectively. From Theorem \ref{theoremPsi}, we deduce that $\Psi=\Psi_{T_0} + \Psi_{T_1}$ is a fouling map for $H$ if and only if 
\begin{enumerate}
    \item[(i)] $d(\flat(T_{0}))=0$ (or equivalently $R_{T_0}=0$),
    \item[(ii)] $R_{{T}_{1}}=0,$
    \item[(iii)] $d(\flat(i_{dV}(T_{1})))=0.$
\end{enumerate}
In coordinates $(\theta,\phi)$, if 
\begin{equation}
T_{0}=A_{0}\partial_\theta+B_{0}\partial_\phi,\quad A_{0},B_{0}\in C^{\infty}(S^{2})
\end{equation}
then 

$$
\flat(T_{0})=A_{0}d\theta+B_0{\sin^2\theta}d\phi.
$$
Now, $d(\flat(T_0))=0$ is equivalent to  
\begin{equation}\label{T0S}
{\partial_\phi}A_0 - \sin^2\theta{\partial_\theta B_0}-B_0\sin (2\theta)=0.
\end{equation}
On the other hand, 
\begin{equation}\label{T1S2}
T_{1}=A_1\partial_\theta\otimes \partial_\theta+B_1Sym(\partial_\phi\otimes \partial_\theta)+C_1\partial_\phi\otimes \partial_\phi,\quad A_1,B_1,C_1\in C^{\infty}(S^{2}).
\end{equation}


$R_{T_{1}}=0$ if and only if
$$\left\{\begin{array}{ll}
{\partial_\theta} A_1=0,&
{\partial_\phi} A_{1}-2\sin\theta\cos\theta B_{1}=0,\\[5pt]
{\partial_\theta} B_{1}+\cot\theta B_{1}=0,&
{\partial_\phi} B_{1}+\cot\theta A_{1}-\sin\theta\cos\theta C_{1}=0,\\[5pt]
{\partial_\theta} C_{1}+2\cot\theta C_{1}=0,&
{\partial_\phi} C_{1}+2\cot\theta B_{1}=0.\\[5pt]
\end{array}
\right.
$$

The solution of this system is $A_{1}=a_{1},B_{1}=0,C_{1}=\displaystyle\frac{a_{1}}{\sin^{2}\theta},$ with $ a_{1} \in\mathbb{R}$.
Finally, let $V\in C^{\infty}(S^{2}).$ Then $\flat(i_{dV}T_{1})=a_1dV$ and 
\begin{equation}
d(\flat(i_{dV}T_{1}))=a_1d(dV)=0,
\end{equation}
which is true for arbitrary $a_{1}$. 

In conclusion,  $\Psi=\Psi_{T_{0}}+\Psi_{T_{1}}$ is a fouling map for the mechanical Hamiltonian function $H=\frac{1}{2}(p_{\theta}^{2}+\frac{1}{\sin^{2}\theta}p_{\phi}^{2})+V(\theta,\phi)$ if and only if the coefficients of the vector field  ${T_0}$ satisfy (\ref{T0S})  and the coefficients of the tensor $T_1$ are $A_1=a_{1}\in\mathbb{R}$, $B_1=0$ and $C_1=\displaystyle\frac{a_{1}}{\sin^{2}\theta}$.

In fact, $$\Psi(\theta,\phi,p_\theta,p_\phi)=(\theta,\phi,A_0(\theta,\phi)+a_1 p_\theta,B_0(\theta,\phi)\sin^2\theta + a_1p_\phi).$$
Now, using (\ref{T0S}), we deduce that 
 $\Psi^{*}\omega_Q=a_{1}\omega_Q$.

\subsection{The case of polynomial maps of degree 2}
\label{SS2}
Now, let $T_{2}$ be a symmetric $(3,0)$-tensor field on $S^2$. From Theorem \ref{theoremPsi}, we deduce that $\Psi=\Psi_{T_0} + \Psi_{T_1}+ \Psi_{T_2}$ is a fouling map for $H$ if and only if 
\begin{enumerate}
    \item[(i)]
    $R_{T_2}=0$,
    \item[(ii)]
    $R_{T_1}=0$,
    \item[(iii)]
$R_{T_0}=i_{dV}T_2,$
    \item[(iv)]
    $d(\flat(i_{dV}(T_{1})))=0.$
\end{enumerate}

In coordinates $(\theta,\phi)$, we have
$$
\begin{array}{rcl}
T_{2}&=&A_{2}\displaystyle\partial_\theta\otimes \partial_\theta\otimes \partial_\theta+B_{2}Sym(\partial_\phi\otimes \partial_\theta\otimes \partial_\theta)+
C_{2}Sym(\partial_\phi\otimes \partial_\phi\otimes \partial_\theta)
+D_{2}\displaystyle\partial_\phi\otimes \partial_\phi\otimes\partial_\phi,
\end{array}$$
with $A_{2},B_{2},C_{2},D_{2}\in C^{\infty}(S^{2})$. The system of partial differential equations for $R_{T_{2}} = 0$ is given by
\begin{equation}
\left\{
\begin{array}{ll}
    {\partial_\theta} A_{2} = 0, &
    {\partial_\phi}A_{2} - 3\sin\theta\cos\theta \, B_{2} = 0, \\
    {\partial_\theta}B_{2} + \cot\theta \, B_{2} = 0, &
    {\partial_\phi}B_{2} - 2\sin\theta\cos\theta \, C_{2} + \cot\theta \, A_{2} = 0, \\
    {\partial_\theta}C_{2} + 2\cot\theta \, C_{2} = 0, &
    {\partial_\phi}C_{2} - \sin\theta\cos\theta \, D_{2} + 2\cot\theta \, B_{2} = 0, \\
    {\partial_\theta}D_{2} + 3\cot\theta \, D_{2} = 0, &
    {\partial_\phi}D_{2} + 3\cot\theta \, C_{2} = 0.
\end{array}\right.
\end{equation}
The only solution to $R_{{T}_{2}}=0$ is $A_{2}=B_{2}=C_{2}=D_{2}=0$, i.e., $T_{2}=0$.

\subsection{The case of polynomial maps  of degree 3}\label{SS3}
Finally, let us consider a symmetric $(4,0)$-tensor field $T_{3}$ on $S^2$. The conditions in Theorem \ref{theoremPsi} reduce to
\begin{enumerate}
    \item[(i)]
    $R_{T_3}=0$ and $R_{T_{2}}=0$,
    \item[(ii)]
    $R_{T_1}=i_{dV}T_{3}$,
    \item[(iii)]
$R_{T_0}=i_{dV}T_2,$
    \item[(iv)]
    $d(\flat(i_{dV}(T_{1})))=0.$
\end{enumerate}
We express $T_{3}$ on spherical coordinates as
$$\begin{array}{rcl}
T_{3}&=&A_{3}\displaystyle\partial_\theta\otimes \partial_\theta\otimes\partial_\theta\otimes\partial_\theta
+B_{3}Sym(\partial_\phi\otimes\partial_\theta\otimes \partial_\theta\otimes\partial_\theta)
+C_{3}Sym(\partial_\phi\otimes \partial_\phi\otimes \partial_\theta\otimes\partial_\theta)\\
&+&D_{3}Sym(\partial_\phi\otimes\partial_\phi\otimes \partial_\phi\otimes\partial_\theta)
+F_{3}\displaystyle\partial_\phi\otimes \partial_\phi\otimes\partial_\phi\otimes\partial_\phi,
\end{array}$$
with $A_{3},B_{3},C_{3},D_{3},F_{3}\in C^{\infty}(S^{2})$.

The system of partial differential equations for $R_{T_{3}} = 0$ is given by
\begin{equation}
\left\{
\begin{array}{ll}
    {\partial_\theta} A_{3} = 0, &
    {\partial_\phi}A_{3} - 4\sin\theta\cos\theta \, B_{3} = 0, \\
    {\partial_\theta}B_{3} + \cot\theta \, B_{3} = 0, &
    {\partial_\phi}B_{3} - 3\sin\theta\cos\theta \, C_{3} + \cot\theta \, A_{3} = 0, \\
    {\partial_\theta}C_{3} + 2\cot\theta \, C_{3} = 0, &
    {\partial_\phi}C_{3} - 2\sin\theta\cos\theta \, D_{3} + 2\cot\theta \, B_{3} = 0, \\
    {\partial_\theta}D_{3} + 3\cot\theta \, D_{3} = 0, &
    {\partial_\phi}D_{3} - \sin\theta\cos\theta \, F_{3} + 3\cot\theta \, C_{3} = 0, \\
    {\partial_\theta}F_{3} + 4\cot\theta \, F_{3} = 0, &
    {\partial_\phi}F_{3} + 4\cot\theta \, D_{3} = 0.
\end{array}\right.
\end{equation}

The solution to $R_{T_{3}}=0$ is $A_{3}=3a_{3},\, B_{3}=0,\, C_{3}=\displaystyle\frac{a_{3}}{\sin^{2}\theta},\, D_{3}=0,\, F_{3}=\displaystyle\frac{3a_{3}}{\sin^{4}\theta}$ with $a_{3}\in\mathbb{R}.$ We rewrite
\begin{equation}\label{eqRT_{3}}
\begin{array}{rcl}
T_{3}&=&3a_{3}\displaystyle\partial_\theta\otimes \partial_\theta\otimes\partial_\theta\otimes\partial_\theta
+\displaystyle\frac{a_{3}}{\sin^{2}\theta}Sym(\partial_\phi\otimes \partial_\phi\otimes \partial_\theta\otimes\partial_\theta)
+\displaystyle\frac{3a_{3}}{\sin^{4}\theta}\displaystyle\partial_\phi\otimes \partial_\phi\otimes\partial_\phi\otimes\partial_\phi.
\end{array}
\end{equation}

Proceeding as in Subsection \ref{SS2}, we deduce that $R_{T_{2}}=0$ implies  $T_{2}=0$. Now, if $T_1$ is defined as (\ref{T1S2}),  the condition $R_{T_{1}}=i_{dV}T_{3}$ is equivalent to the following system of differential equations
\begin{equation}\label{eqT_{1}T_{3}}
\left\{
\begin{array}{l}
    {\partial_\theta}A_{1} = 3a_{3}{\partial_\theta}V,\\[5pt]
    \dfrac{1}{\sin^2\theta} \left({\partial_\phi}A_{1} - 2\sin\theta\cos\theta \, B_{1} \right) = \dfrac{a_{3}}{\sin^{2}\theta} {\partial_\phi}V, \\[5pt]
    {\partial_\theta}C_{1} + 2\cot\theta \, C_{1} = \dfrac{a_{3}}{\sin^{2}\theta} {\partial_\theta}V, \\[5pt]
    \dfrac{1}{\sin^2\theta} \left({\partial_\phi}C_{1} + 2\cot\theta \, B_{1}\right) = \dfrac{3a_{3}}{\sin^{4}\theta} {\partial_\phi}V, \\[5pt]
    \dfrac{2}{\sin^2\theta} \left({\partial_\phi}B_{1} + \cot\theta \, A_{1} - \sin\theta\cos\theta \, C_{1} \right) - {\partial_\theta}C_{1} -2\cot\theta \, C_{1} = \dfrac{a_{3}}{\sin^{2}\theta} {\partial_\theta}V, \\[5pt]
    2{\partial_\theta}B_{1} + 2\cot\theta \, B_{1} - \dfrac{1}{\sin^2\theta} \left( {\partial_\phi}A_{1} - 2\sin\theta\cos\theta \, B_{1} \right) = \dfrac{a_{3}}{\sin^{2}\theta}{\partial_\phi}V.
\end{array}\right.
\end{equation}

Finally, let us see the condition $d(\flat(i_{dV}T_{1}))=0$, it is equivalent to
\begin{equation}\label{eqi_{dV}T_{1}}
\partial_{\theta}(B_{1}\partial_{\theta}V+C_{1}\partial_{\phi}V)\sin^{2}\theta+2(B_{1}\partial_{\theta}V+C_{1}\partial_{\phi}V)\sin\theta \cos\theta-\partial_{\phi}(A_{1}\partial_{\theta}V+B_{1}\partial_{\phi}V)=0.
\end{equation}

In conclusion, $\Psi_T=\Psi_{T_{0}}+\Psi_{T_{1}}+\Psi_{T_{2}}+\Psi_{T_{3}}$ is a fouling map for the mechanical Hamiltonian function $H=\frac{1}{2}(p_{\theta}^{2}+\displaystyle \frac{1}{\sin^{2}\theta}p_{\phi}^{2})+V(\theta,\phi)$, if and only if the coefficients of $T_0$ satisfy (\ref{T0S}), $T_{2}=0$, $T_{3}$ is of the form (\ref{eqRT_{3}}), and $T_{1}$ satisfies equations (\ref{eqT_{1}T_{3}}) and (\ref{eqi_{dV}T_{1}}).

\begin{example}
\label{example6.1}
Let $V(\theta,\phi)=\sin^{2}\theta$. System (\ref{eqT_{1}T_{3}}) reduces to
\begin{equation}
\left\{
\begin{array}{ll}
    {\partial_\theta}A_{1} = 6 a_{3}\sin\theta\cos\theta, &
    {\partial_\phi}A_{1} - 2\sin\theta\cos\theta \, B_{1} = 0, \\
    {\partial_\theta}B_{1} + \cot\theta \, B_{1} = 0, &
    {\partial_\phi}B_{1} + \cot\theta \, A_{1} - \sin\theta\cos\theta \, C_{1} = 2 a_{3} \sin\theta\cos\theta, \\
    {\partial_\theta}C_{1} + 2\cot\theta \, C_{1} = 2 a_{3}\cot\theta, &
    {\partial_\phi}C_{1} + 2\cot\theta \, B_{1} = 0,
\end{array}\right.
\end{equation}
whose solution is $A_{1}=a_1+3a_{3}\sin^{2}\theta, \, B_{1}=0, \, C_{1}=\displaystyle\frac{a_1}{sin^{2}\theta}+a_{3}$. 

On the other hand, since $B_1=0$ equation (\ref{eqi_{dV}T_{1}}) reduces to
\begin{equation}
\partial_{\phi} A_{1}=0,
\end{equation}
which is satisfied by the previous solution. 

Then,  we have that for arbitrary $a_1,a_{3}$, the map
\begin{equation}
\begin{split}
&\Psi_T(\theta,\phi,p_{\theta},p_{\phi})=(\theta,\phi,P_{\theta},P_{\phi}),\\
&P_{\theta}=A_0+\big(a_1+3a_3\sin^2\theta\big)p_\theta+\frac{a_3}{2}\,p_\theta\Big(p_\theta^2+\frac{p_\phi^2}{\sin^2\theta}\Big),\\
&P_{\phi}=B_0+\big(a_1+a_3\sin^2\theta\big)p_\phi+\frac{a_3}{2}\,p_\phi\Big(p_\theta^2+\frac{p_\phi^2}{\sin^2\theta}\Big),
\end{split}
\end{equation}
with ${\partial_\phi}A_0 - \sin^2\theta{\partial_\theta B_0}-B_0\sin (2\theta)=0$, is a fouling map for $H=\displaystyle\frac{1}{2}(p_{\theta}^{2}+\frac{1}{\sin^{2}\theta}p_{\phi}^{2})+\sin^{2}\theta$.

The associated constants of motion (up to the power $2$ of the tensor $L_{\Psi_T}$) are
\begin{equation}
f_{1}=4\big(a_1+2a_3H\big)
\end{equation}
and
\begin{equation}
f_{2}=\displaystyle\frac14f_1^{\,2}+4a_3^2\big(H^2-2p_\phi^2\big),
\end{equation}
which are functionally independent for $a_3\neq 0$. 
\end{example}

\begin{example}
Now let us consider the potential $V(\theta,\phi)=1-\sin^{2}\theta\cos^{2}\phi$, or equivalently $V(\theta,\phi)=\sin^{2}\phi+\cos^{2}\theta\cos^{2}\phi$.

The solution to system (\ref{eqT_{1}T_{3}}) and equation (\ref{eqi_{dV}T_{1}}) is $A_{1}=a_1 +a_3 \left(3\cos^{2}\theta\cos^{2}\phi+\sin^{2}\phi\right)$, $B_{1}=-2a_3 \cot\theta\,\sin\phi\cos\phi$ and $C_{1}=a_1 \csc^{2}\theta+a_3 \csc^{2}\theta\left(3\sin^{2}\phi+\cos^{2}\theta\cos^{2}\phi\right)$.

So, we have that for arbitrary $a_{1},a_{3}\in {\Bbb R}$, the map
\begin{equation}
\begin{split}
&\Psi_T(\theta,\phi,p_{\theta},p_{\phi})=(\theta,\phi,P_{\theta},P_{\phi}),\\
&P_{\theta}=A_{0}+\left[a_1 +a_3 \left(3\cos^{2}\theta\cos^{2}\phi+\sin^{2}\phi\right)\right]p_{\theta}-2a_3 \cot\theta\,\sin\phi\cos\phi\,p_{\phi}+\frac{a_3}{2}\,p_{\theta}\left(p_{\theta}^{2}+\csc^{2}\theta\,p_{\phi}^{2}\right)\\
&P_{\phi}=B_{0}-2a_3 \sin\theta\cos\theta\,\sin\phi\cos\phi\,p_{\theta}+\left[
a_1 +a_3 \left(3\sin^{2}\phi+\cos^{2}\theta\cos^{2}\phi\right)\right]p_{\phi}+\frac{a_3}{2}\,p_{\phi}\left(p_{\theta}^{2}+\csc^{2}\theta\,p_{\phi}^{2}\right),
\end{split}
\end{equation}
with ${\partial_\phi}A_0 - \sin^2\theta{\partial_\theta B_0}-B_0\sin (2\theta)=0$, is a fouling map for $H=\displaystyle\frac{1}{2}(p_{\theta}^{2}+\frac{1}{\sin^{2}\theta}p_{\phi}^{2})+\sin^{2}\phi+\cos^{2}\theta\cos^{2}\phi$.

The associated constants of motion (up to the power $2$ of the tensor $L_{\Psi_T}$) are
\begin{equation}
f_{1}=4(a_1 +2a_3 H),
\end{equation}
and
\begin{equation}\label{eq:f2ex62}
  f_2=\tfrac14 f_1^{\,2}+4a_3^{2}\bigl(D^{2}+E^{2}\bigr),
\end{equation}
where
\begin{align}\label{eq:DEex62}
  D&=\cos^{2}\theta\cos^{2}\phi-\sin^{2}\phi
     +\tfrac12\Bigl(p_\theta^{2}-\frac{p_\phi^{2}}{\sin^{2}\theta}\Bigr),\\[2pt]
  E&=\frac{p_\theta\,p_\phi}{\sin\theta}-2\cos\theta\sin\phi\cos\phi .
\end{align}
Moreover, if $J_1=-\bigl(\sin\phi\,\partial_\theta+\cot\theta\cos\phi\,\partial_\phi\bigr)$
denotes the Killing vector field generating the rotations about the $x_1$-axis,
$x_1=\sin\theta\cos\phi$, and
\begin{equation}\label{eq:Jex62}
  J_1^{\ell}=-\bigl(\sin\phi\,p_\theta+\cot\theta\cos\phi\,p_\phi\bigr)
\end{equation}
is the corresponding linear function on $T^{*}S^{2}$, then
\begin{equation}\label{eq:DEidex62}
  D^{2}+E^{2}=H^{2}-2\bigl(J_1^{\ell}\bigr)^{2},
\end{equation}
so that
\begin{equation}\label{eq:f2bisex62}
  f_2=\tfrac14 f_1^{\,2}+4a_3^{2}\Bigl(H^{2}-2\bigl(J_1^{\ell}\bigr)^{2}\Bigr).
\end{equation}
Thus $f_1$ and $f_2$ are functionally independent for $a_3\neq0$.
\end{example}

\begin{remark}
To the best of our knowledge, in the existing literature, we are not aware of any non-trivial fouling (or canonoid) transformation for a natural Hamiltonian system on a curved space. In the previous examples, we are presenting non-trivial fouling maps for natural Hamiltonian systems on the 2-sphere. It is worth remarking that the dynamics considered are not geodesic, in fact, the corresponding potential functions are non-trivial; in fact, they are degenerate quadratic Neumann potentials \citep{ractiu1981c,dullin2012degenerate}.
\end{remark}

\section{Fouling maps for mechanical Hamiltonian functions on a $2$-dimensional Riemannian  manifold with a Liouville metric}
\label{sec6}

Let us consider a generic $2$-dimensional Riemannian manifold $(M,g)$. The metric $g$ is called a Liouville metric \citep{bolsinov1998two} if it can be written in suitable local coordinates $(x, y)$ as
\begin{equation}
g(x,y)= \bigl(f(x) + h(y)\bigr)(dx^2 + dy^2).
\end{equation}
We take the open manifold $M=\{(x,y)\in{\Bbb R}^2/x+y>0\}$, the Liouville metric 
$$g=(x+y)(dx^2+dy^2)$$
and the potential function $V(x,y)=3x^2+2xy+3y^2$. We make the coordinate transformation $u=x+y$ and  $v=x-y$, then
$g=\displaystyle\frac{u}{2}(du^2+dv^2)$,
\begin{equation}
g^*=\frac{2}{u}\left(\partial_u\otimes\partial_u+
\partial_v\otimes\partial_v\right)
\end{equation}
and $V=2u^2+v^2$.

Let $T_3$ be  a symmetric $(4,0)$-tensor on $M$ such that $R_{T_3}=0$. Then 

\begin{equation}
\begin{array}{ll}
T_3(du,du,du,du)=\displaystyle\frac{12 a_3}{u^2},&
T_3(du,du,dv,dv)=\displaystyle\frac{4 a_3}{u^2},\\[8pt]
T_3(dv,dv,dv,dv)=\displaystyle\frac{12 a_3}{u^2},&
T_3(du,du,du,dv)=T_3(du,dv,dv,dv)=0.
\end{array}
\end{equation}
Moreover,  $R_{T_2}=0$ implies that  $T_2=0$. The solution for equations $R_{T_1}=i_{dV}T_3$ and $d(\flat(i_{dV}T_{1}))=0$ is
\begin{equation}
\begin{aligned}
T_1(du,du)&=12a_3 u+\frac{4a_3 v^2+a_1}{u},\\
T_1(du,dv)&=4a_3 v,\\
T_1(dv,dv)&=4a_3 u+\frac{4a_3 v^2+a_1}{u}.
\end{aligned}
\end{equation}

Finally, the condition $R_{T_0}=i_{dV}T_2=0,$ implies that 
$$\partial_u (uB_{0})-\partial_v (uA_{0})=0,$$
where 
$T_{0}=A_{0}\partial_u+B_{0}\partial_v,$ with $A_{0},B_{0}\in C^{\infty}(M).$

So we have that, for arbitrary  $a_{1},a_{3}$, the map
\begin{equation}
\begin{split}
&\Psi_T(u,v,p_u,p_v)=(u,v,P_u,P_v),\\
&P_u=A_{0}+\big(\frac{a_1}{2}+6a_3 u^2+2a_3 v^2\big)p_u+\big(2a_3 uv\big)p_v+\frac{a_3}{u}(p_u^2+p_v^2)p_u,\\
&P_v=B_{0}+\big(2a_3 uv\big)p_u+\big(\frac{a_1}{2}+2a_3 u^2+2a_3 v^2\big)p_v+\frac{a_3}{u}(p_u^2+p_v^2)p_v,
\end{split}
\end{equation}
with $B_0+u(\partial_u B_{0}-\partial_v A_{0})=0$, is a fouling map for $H=\displaystyle\frac{1}{u}(p_{u}^{2}+p_{v}^{2})+2u^{2}+v^{2}$.

The associated constants of motion (up to the power $2$ of the tensor $L_{\Psi_T}$) are
\begin{equation}
f_{1}=a_1+4a_3
\left(
\frac{p_u^2+p_v^2}{u}+2u^2+v^2
\right)=a_1+4a_3H,
\end{equation}
and
\begin{equation}
\begin{aligned}
f_{2}&=\frac{a_1^2}{2}+
\frac{4a_1a_3}{u}(p_u^2+p_v^2)+
8a_1a_3 u^2+
4a_1a_3 v^2+
\frac{10a_3^2}{u^2}(p_u^2+p_v^2)^2+
40a_3^2up_u^2+
24a_3^2up_v^2\\
&+
\frac{16a_3^2v^2}{u}(p_u^2+p_v^2)+
16a_3^2v\,p_up_v+
40a_3^2u^4+
40a_3^2u^2v^2+
8a_3^2v^4\\
&=\frac{1}{2}\left(a_1+4a_3H\right)^2
+
\frac{2a_3^2}{u^2}
\left[
\left(p_u^2-p_v^2+2u^3\right)^2
+
\left(2p_up_v+2u^2v\right)^2
\right],
\end{aligned}
\end{equation}
which are functionally independent if $a_3\not=0$.

\section*{Acknowledgements}

This research was initiated while R. Azuaje was at Universidad Autónoma Metropolitana Unidad Iztapalapa in Mexico city, with the financial support of the Secretaría de Ciencia, Humanidades, Tecnología e Innovación (SECIHTI) of Mexico through a postdoctoral fellowship under the Estancias Posdoctorales por México 2022 program. The author received subsequent support for the completion of this research from the European Union and the Czech Ministry of Education under project CZ.02.01.01/00/22\_011/0008569 "Czech Technical University - International Postdoc Programme CROP".

J.C. Marrero and E. Padr\'on acknowledge financial support from the Spanish Ministry of Science and Innovation under grant PID2022-137909NB-C22. 


\section*{AI Use Statement}
The author R. Azuaje acknowledges the use of Artificial intelligence (AI)-assisted tools for the preparation of the examples presented in the paper; in particular for performing simplifications of some algebraic and symbolic expressions, and solving some systems of differential equations as well. All results obtained with the assistance of these tools were independently reviewed and verified by the author. The author takes full responsibility for the accuracy, originality, and integrity of the final manuscript.

\appendix

\section{The proof of Proposition \ref{WT}}\label{Appendix}

In this appendix we will show the proof of Proposition \ref{WT}. 

{\bf \noindent Proposition \ref{WT}. }{\it 
Let $T_k$ be a symmetric $(k+1,0)$-tensor field on $Q,$ with $k\geq  1$. If   $H:T^*Q\to \mathbb{R}$ is a mechanical Hamiltonian function, with kinetic energy induced by the semi-Riemannian metric $g$ and  with potential function $V:Q\to \mathbb{R}$ then, for all $f\in C^\infty(Q)$ and $Y\in {\mathfrak X}(Q),$
\begin{enumerate}
\item[(a)] $\omega_{T_k}(X_H, X_{f\circ \pi_Q})=k(i_{df}{T_k})^{hom}=kX_{f\circ \pi_Q}({T_k}^{hom}),$
\item[(b)]
$\omega_{T_k}(X_H,X_{Y^l})= (k{\mathcal L}_YT_k -i_{\flat(Y)}R_{T_k})^{hom} +(i_{dV}i_{\flat(Y)}{T_k})^{hom},$
\end{enumerate}
where $\omega_{T_k}=\Psi_{T_k}^*(\omega_Q)$ and  $R_{T_k}$ is the $(k+2,0)$-tensor given by 
\begin{equation}
R_{T_k}(\alpha_0,\alpha_1,\dots,\alpha_{k+1})=\sum_{i=1}^{k+1} (\nabla_{\sharp\alpha_i}{T_k})(\alpha_1,\dots,\alpha_{i-1} ,\alpha_{0},\alpha_{i+1},\dots, \alpha_{k+1})-(\nabla_{\sharp\alpha_{0}}{T_k})(\alpha_1, \dots ,\alpha_{k+1}).
\end{equation}
Here $\nabla$ denotes the Levi-Civita connection of $(Q,g).$ }
\begin{proof}
Let $(q^i)$ be coordinates on $Q$ and $(q^i,p_i)$ the corresponding coordinates  on $T^*Q.$  Then, if the local expression of $T$  is defined by the local functions $T^{i_1\dots i_{k+1}}=T_k(dq^{i_1},\dots, dq^{i_{k+1}})$  on $Q,$ the corresponding bundle map is given by 
$$\Psi_{T_k}(q^i,p_i)=(q^i,\frac{1}{k!}g_{ij}T^{i_1\dots i_kj}\prod_{s=1}^kp_{i_s})=(q^i, \frac{1}{k!}\widetilde{T}_i^{i_1\dots i_k}\prod_{s=1}^kp_{i_s}),$$
where $\widetilde{T}_i^{i_1\dots i_{k}}=g_{ij}T^{i_1\dots i_{k}j}$ and $\displaystyle\prod_{s=1}^kp_{i_s}$ denotes the homogeneous polynomial of degree $k$ 
$$\prod_{s=1}^kp_{i_s}=p_{i_1}\dots p_{i_k}.$$
The functions  $\widetilde{T}_i^{i_1\dots i_k}$ are the coefficients of the $(k,1)$-tensor  $\widetilde{T}_k$ deduced from the $(k+1,0)$-tensor $T_k$ as in (\ref{Tt}).
Therefore, 
\begin{equation}
\label{wT}
\omega_{T_k}=\displaystyle\frac{1}{k!} \partial_{q^j} (\widetilde{T}_i^{i_1\dots i_k})\prod_{s=1}^kp_{i_s}dq^i\wedge dq^j + \frac{1}{(k-1)!}\widetilde{T}_i^{i_1\dots i_k}\prod_{s=2}^kp_{i_s} dq^i\wedge dp_{i_1}.
\end{equation}

$(a)$ First we will prove that 
     \begin{equation}\label{10}
     \omega_{T_k}(X_H,X_{f\circ \pi_Q})=k(i_{df}{T_k})^{hom}.
     \end{equation}
Note that we have supposed that ${T_k}$ is symmetric and therefore $i_{df}{T_k}$ is symmetric as well.

Now,  using (\ref{hamiltonian}) and (\ref{wT}), we have that 
\begin{equation}\label{2}
\langle i_{X_H}\omega_{T_k}, \partial_{p_j}\rangle=\frac{1}{(k-1)!}  T^{i_1i_2\dots i_{k}j}\prod_{s=1}^kp_{i_s}.
\end{equation}
So, from (\ref{2.0}) and since $X_{f\circ \pi_Q}=\partial_{q^j}f\partial_{p_j},$  we deduce (\ref{10}). In addition, using item $(a)$ in Proposition \ref{lemma1}, we have that 
$$\omega_{T_k}(X_H,X_{f\circ \pi_Q})=k(i_{df}{T_k})^{hom}=kX_{f\circ \pi_Q}({T_k}^{hom}).$$ 

$(b)$ From (\ref{hamiltonian}) and (\ref{wT}),
\begin{equation*}
\begin{array}{rcl}
\langle i_{X_H}\omega_{T_k}, \partial_{q^j}\rangle &=& g^{i_{k+1}r} \left(\displaystyle\frac{1}{k!}\left(\partial_{q^j}\widetilde{T}_r^{i_1 \dots i_k}-\partial_{q^r}\widetilde{T}_j^{i_1\dots i_k}\right)-\displaystyle\frac{1}{(k-1)!}\widetilde{T}_j^{li_2\dots i_k}\Gamma_{rl}^{i_1}\right)(\displaystyle\prod_{s=1}^{k+1}p_{i_s}) \\
&+& \displaystyle\frac{1}{(k-1)!}\widetilde{T}_j^{i_1\dots i_k}\partial_{q^{i_1}}V\prod_{s=2}^kp_{i_s}
\end{array}
\end{equation*}
and therefore, using (\ref{24'}) and (\ref{2}), we deduce that 
$$\begin{array}{rcl}\omega_{T_k}({X_H}, X_{Y^\ell})&=& \left(\displaystyle\frac{1}{k!}Y^jg^{i_{k+1}r} \left(\partial_{ q^j}\widetilde{T}_r^{i_1\dots i_k}-\partial_{q^r}\widetilde{T}_j^{i_1\dots i_k}- k\widetilde{T}^{li_2\dots i_k}_j\Gamma_{rl}^{i_{1}}\right)\right.
\\[12pt]&&-\left.\displaystyle\frac{1}{(k-1)!}g^{js}\widetilde{T}_s^{i_1\dots i_{k}} \partial_{q^j}Y^{i_{k+1}}\right)\displaystyle\prod_{s=1}^{k+1}p_{i_s} +\displaystyle\frac{1}{(k-1)!}Y^j\widetilde{T}_j^{i_1\dots i_k}\partial_{ q^{i_1}}V\prod_{s=2}^kp_{i_s}.\end{array}$$
Now, 
    $$\langle i_{dV}\widetilde{T_k},Y\rangle^{hom}=\frac{1}{(k-1)!}Y^j\widetilde{T}_j^{i_1\dots i_k}\partial_{q^{i_1}}V\prod_{s=2}^kp_{i_s}.$$

On the other hand, 
$$\langle (\nabla_{\sharp dq^{i_{k+1}}}\widetilde{T_k})(dq^{i_1},\dots, dq^{i_k}),Y\rangle=
Y^jg^{i_{k+1}r}\left(\partial_{q^r}\widetilde{T}_j^{i_1\dots i_k}-\widetilde{T}^{i_1\dots i_k}_l\Gamma_{rj}^l+\Gamma^{i_m}_{rl}\widetilde{T}_j^{li_1\dots i_{m-1}i_{m+1}\dots i_k}\right)$$

$$\langle({\mathcal L}_Y\widetilde{T_k})
(dq^{i_1},\dots ,dq^{i_k}),\sharp dq^{i_{k+1}}\rangle=g^{i_{k+1}r}\left(Y^j\partial_{q^j}\widetilde{T}_{r}^{i_1\dots i_k}+\widetilde{T}_{j}^{i_1\dots i_k}\partial_{q^{r}}Y^{j}-\partial_{q^{l}}Y^{i_m}\widetilde{T}_r^{li_1\dots i_{m-1}i_{m+1}\dots  i_k}\right).$$

$$\langle \widetilde{T_k}(dq^{i_1},\dots, dq^{i_k}), \nabla_{\sharp dq^{i_{k+1}}} Y\rangle=\widetilde{T}^{i_1\dots i_{k}}_ig^{{i_{k+1}}r}(\partial_{q^r}Y^i+Y^j\Gamma_{rj}^i).$$
Therefore, 
\begin{equation}\label{HY}\begin{array}{rcl}\omega_{T_k}({X_H}, X_{Y^\ell})&=&\displaystyle\frac{1}{k!}\left(-\langle (\nabla_{\sharp dq^{i_{k+1}}}(\widetilde{T_k})(dq^{i_1},\dots, dq^{i_k}),Y\rangle +\langle({\mathcal L}_Y\widetilde{T_k})
(dq^{i_1},\dots ,dq^{i_k}),\sharp dq^{i_{k+1}}\rangle \right.\\[5pt]&&\left.-\langle \widetilde{T_k}(dq^{i_1},\dots, dq^{i_k}), \nabla_{\sharp dq^{i_{k+1}}} Y\rangle\right )\displaystyle\prod_{s=1}^{k+1}p_{i_s} + \langle i_{dV}\widetilde{T_k},Y\rangle^{hom}
\end{array}\end{equation}

Now, we will describe 
$$S_Y(\alpha_1,\dots ,\alpha_{k+1}):=-\langle (\nabla_{\sharp \alpha_{k+1}}\widetilde{T_k})(\alpha_1,\dots, \alpha_k),Y\rangle +\langle({\mathcal  L}_Y\widetilde{T_k})(\alpha_1,\dots ,\alpha_k),\sharp \alpha_{k+1}\rangle - \langle \widetilde{T_k}(\alpha_1,\dots, \alpha_k), \nabla_{\sharp \alpha_{k+1}} Y\rangle$$
in terms of the tensor $T_k.$

We have that 
$$
\begin{array}{rcl}
({\mathcal L}_YT_k)(\alpha_1,\dots \alpha_{k+1})&=&Y({T_k}(\alpha_1,\dots, \alpha_k,\alpha_{k+1})) -\displaystyle\sum_{i=1}^{k+1} {T_k}(\alpha_1,\dots, {\mathcal L}_Y\alpha_i,\dots, \alpha_{k+1})=Y(\langle \widetilde{T_k}(\alpha_1,\dots, \alpha_k), \sharp \alpha_{k+1}\rangle\\&& -\displaystyle\sum_{i=1}^k\langle \widetilde{T_k}(\alpha_1,\dots, {\mathcal L}_Y\alpha_i,\dots, \alpha_k),\sharp \alpha_{k+1}\rangle-T_k(\alpha_1,\dots, \alpha_k,{\mathcal L}_Y\alpha_{k+1})\\[5pt]&=& \langle ({\mathcal L}_Y\widetilde{T_k})(\alpha_1,\dots, \alpha_k),\sharp \alpha_{k+1}\rangle-\langle {\mathcal L}_Y(\widetilde{T_k}(\alpha_1,\dots, \alpha_k)),\sharp \alpha_{k+1}\rangle + Y(\langle \widetilde{T_k}(\alpha_1,\dots, \alpha_k),\sharp \alpha_{k+1}\rangle)\\[5pt]&&-T_k(\alpha_1,\dots, \alpha_k,{\mathcal L}_Y\alpha_{k+1})\\[5pt] &=& \langle ({\mathcal L}_Y\widetilde{T_k})(\alpha_1,\dots, \alpha_k),\sharp \alpha_{k+1}\rangle +\langle \widetilde{T_k}(\alpha_1,\dots, \alpha_k),{\mathcal L}_Y\sharp \alpha_{k+1}\rangle -T_k(\alpha_1,\dots, \alpha_k,{\mathcal L}_Y\alpha_{k+1}),
\end{array}$$

with $\alpha_i\in \Omega^1(Q).$ Here we have used that 
$$\langle {\mathcal L}_Y\alpha,Z\rangle + \langle \alpha,{\mathcal L}_YZ\rangle=Y\langle \alpha,Z\rangle, \;\;  \mbox{ for all }Y,Z\in {\mathfrak X}(Q)\mbox{ and } \alpha\in \Omega^1(Q).$$

Thus,
$$
\begin{array}{rcl}
\langle {\mathcal L}_Y\widetilde{T_k}(\alpha_1,\dots, \alpha_k),\sharp \alpha_{k+1}\rangle- \langle \widetilde{T_k}(\alpha_1,\dots, \alpha_k),\nabla_{\sharp\alpha_{k+1}}Y\rangle&\kern-10pt=&\kern-10pt({\mathcal L}_YT_k)(\alpha_1,\dots \alpha_{k+1})-\langle \widetilde{T_k}(\alpha_1,\dots, \alpha_k),{\mathcal L}_Y\sharp \alpha_{k+1}\rangle \\[5pt]&&\kern-35pt+T_k(\alpha_1,\dots, \alpha_k,{\mathcal L}_Y\alpha_{k+1})-\langle\widetilde{T_k}(\alpha_1,\dots, \alpha_k),\nabla_{\sharp\alpha_{k+1}}Y\rangle
\\[5pt] &\kern-110pt =&\kern-60pt
({\mathcal L}_YT_k)(\alpha_1,\dots, \alpha_{k+1})
+{T_k}(\alpha_1,\dots, \alpha_k,{\mathcal L}_Y \alpha_{k+1}-\nabla_Y {\alpha_{k+1}})
\end{array}
$$

Note that ${\mathcal L}_Y\sharp\alpha_{k+1}=\nabla_Y\sharp\alpha_{k+1}-\nabla_{\sharp\alpha_{k+1}}Y.$

On the other hand,
\begin{equation}\label{RR1}(\nabla_{\sharp\alpha_{k+1}}{T_k})(\alpha_{1},\dots, \alpha_{k},\flat(Y))
=\langle ({\nabla}_{\sharp\alpha_{k+1}}\widetilde{T_k})(\alpha_{1},\dots, \alpha_{k}),Y\rangle.\end{equation}
In fact, using Lemma \ref{nablasos} (b), 
$$
\begin{array}{rcl}
(\nabla_{\sharp\alpha_{k+1}}{T_k})(\alpha_{1},\dots, \alpha_{k},\flat(Y))&=&\sharp\alpha_{k+1}(T_k(\alpha_{1},\dots, \alpha_{k},\flat(Y)))-\displaystyle\sum_{i=1}^{k}{T_k}(\alpha_{1},\dots, \nabla_{\sharp\alpha_{k+1}}\alpha_i,\dots, \alpha_{k},\flat(Y))\\[5pt]&&-{T_k}(\alpha_{1},\dots, \alpha_{k},\nabla_{\sharp\alpha_{k+1}}\flat(Y))\\[5pt]&=&\sharp\alpha_{k+1}\langle \widetilde{T_k}(\alpha_{1},\dots, \alpha_{k}),Y\rangle-\displaystyle\sum_{i=1}^{k}\langle\widetilde{T_k}(\alpha_{1},\dots, \nabla_{\sharp\alpha_{k+1}}\alpha_i,\dots, \alpha_{k}),Y\rangle\\[5pt]&&-\langle\widetilde{T_k}(\alpha_{1},\dots, \alpha_{k}),\nabla_{\sharp\alpha_{k+1}}Y\rangle=
\langle(\nabla_{\sharp\alpha_{k+1}}\widetilde{T_k})(\alpha_{1},\dots, \alpha_{k}),Y\rangle.
\end{array}$$


Therefore,
$$
\begin{array}{l}S_Y(\alpha_1,\dots ,\alpha_{k+1})= (\nabla_{\sharp\alpha_{k+1}}{T_k})(\alpha_{1},\dots, \alpha_{k},\flat(Y)) + ({\mathcal L}_YT_k)(\alpha_1,\dots, \alpha_{k+1})
+{T_k}(\alpha_1,\dots, \alpha_k,{\mathcal L}_Y \alpha_{k+1}-\nabla_Y {\alpha_{k+1}})\end{array}$$
and, since ${\mathcal L}_Y-\nabla_Y:\Omega^1(Q)\to \Omega^1(Q)$  is a $(1,1)$-tensor, then $S_Y$
is a $(k+1,0)$-tensor. But this tensor, in general, is not symmetric. But, its symmetrization $Sym(S_Y)$ satisfies 
\begin{equation}\label{Sym}
(Sym(S_Y))^{hom}(\alpha)=\frac{1}{(k+1)!}S_Y(\alpha,\dots ,\alpha)
\end{equation}

A direct computation proves that 
$$
\begin{array}{rcl}
Sym (S_Y)(\alpha_1,\dots, \alpha_{k+1})&=&-\displaystyle\frac{1}{k+1}\sum_{i=1}^{k+1} (\nabla_{\sharp\alpha_{i}}{T_k})(\alpha_{1},\dots,\alpha_{i-1},\flat(Y),\alpha_{i+1},\dots, \alpha_{k+1})\\&&+({\mathcal L}_YT_k)(\alpha_1,\dots, \alpha_{k+1})
+\displaystyle\frac{1}{k+1}\sum_{i=1}^{k+1}{T_k}(\alpha_1,\dots, {\mathcal L}_Y \alpha_{i}-\nabla_Y \alpha_{i},\dots, \alpha_{k+1}) \\[5pt]&=&
\displaystyle\frac{1}{k+1}\left(-R_{T_k}\left(\flat(Y),\alpha_1,\dots, \alpha_{k+1}\right)+ k({\mathcal L}_YT_k)(\alpha_1,\dots, \alpha_{k+1})\right)
\end{array}
$$

Finally, using (\ref{2.0}), (\ref{HY}) and (\ref{Sym}),  we conclude that 
$$\begin{array}{rcl}\omega_{T_k}({X_H}, X_{Y^\ell})&=&-(i_{\flat(Y)}R_{T_k})^{hom} + k({\mathcal L}_YT_k)^{hom} +  (i_{\flat(Y)}i_{dV}{T_k})^{hom}.
\end{array}$$
\end{proof}

\bibliography{refs} 
\bibliographystyle{plain} 

\end{document}